\documentclass[lettersize,journal,twoside]{IEEEtran}
\usepackage{amsmath,amssymb,amsfonts}
\usepackage{array}
\usepackage[caption=false,font=normalsize,labelfont=sf,textfont=sf]{subfig}
\usepackage{textcomp}
\usepackage{stfloats}
\usepackage{url}
\usepackage{verbatim}
\usepackage{graphicx}
\usepackage{cite}
\usepackage{color}
\usepackage{setspace}
\usepackage{subfig}
\allowdisplaybreaks[1]
\usepackage[linesnumbered, ruled]{algorithm2e}
\SetKwRepeat{Do}{do}{while}%
\usepackage[thmmarks,amsmath]{ntheorem}
\theoremseparator{:}

\newtheorem{lemma}{Lemma}
\newtheorem{proposition}{Proposition}
\theoremheaderfont{\it}
\theorembodyfont{\normalfont}
\theoremsymbol{\ensuremath{\Box}}
\newtheorem*{proof}{Proof}
\theoremheaderfont{\it}
\theorembodyfont{\normalfont}
\theoremsymbol{}
\newtheorem{remark}{Remark}
\usepackage{setspace}
\begin{document}

\title{Orchestrating Communication, Computing, and Energy Transfer for Wireless-Powered 6G Closed-Loop Controls}

\author{Chengleyang Lei, Wei Feng,~\IEEEmembership{Senior Member,~IEEE}, Yanmin Wang, Yunfei Chen,~\IEEEmembership{Fellow,~IEEE}, \\Xiaoyu Liu, Liuguo Yin,~\IEEEmembership{Senior Member,~IEEE}, and Ning Ge,~\IEEEmembership{Member,~IEEE}
	\thanks{Chengleyang Lei, Wei Feng, and Ning Ge are with the Department of Electronic Engineering, State Key Laboratory of Space Network and Communications, Tsinghua University, Beijing 100084, China (email: lcly21@mails.tsinghua.edu.cn; fengwei@tsinghua.edu.cn; gening@tsinghua.edu.cn).}
	\thanks{Yanmin Wang is with the School of Information Engineering, Minzu University of China, Beijing 100081, China~(email: wangyanmin@muc.edu.cn).}
	\thanks{Yunfei Chen is with the Department of Engineering, University of Durham, DH1 3LE Durham, U.K. (e-mail: yunfei.chen@durham.ac.uk).}
	\thanks{Xiaoyu Liu is with the Radio Research Center, China Academy of Information and Communications Technology, Beijing 100083, China (e-mail: liuxiaoyu1@caict.ac.cn).}
	\thanks{Liuguo Yin is with Beijing National Research Center for Information
		Science and Technology, Tsinghua University, Beijing 100084, China (e-mail:
		yinlg@tsinghua.edu.cn).}
	
}




\maketitle
\pagestyle{empty}
\thispagestyle{empty}

\begin{abstract}
Future sixth generation (6G) communications are expected to support robotic control tasks in applications such as industrial automation and emergency response, where sensors, computing units, and robots are interconnected via nervous system-like networks to form sensing-communication-computing-control ($\textbf{SC}^3$) closed loops. However, the limited battery capacities of devices within these $\textbf{SC}^3$ loops constrain operational duration and degrade control efficiency, particularly in remote or post-disaster scenarios. To address this challenge, wireless power transfer (WPT) can be leveraged to provide continuous energy supply for $\textbf{SC}^3$ closed loops. In this paper, we investigate a wireless-powered $\textbf{SC}^3$ system, where a satellite transfers energy via radio frequency (RF) signals to support the communication and computing processes of multiple $\textbf{SC}^3$ closed loops. By accounting for the intricate coupling among computing, communication, and energy transfer, we propose a holistic design framework to enhance overall control performance. Specifically, we adopt the linear quadratic regulator (LQR) cost as the performance metric and formulate a sum LQR cost minimization problem. The uplink/downlink transmit power, bandwidth allocation, computing capability, communication/computing time allocation, and WPT power allocation are jointly optimized. We recast the problem into a more tractable form and develop an iterative algorithm to solve it. For the special case of a single loop, we further analyze the properties of optimal solutions in energy-limited scenarios to provide insights for practical parameter configuration. Simulation results demonstrate the performance gains of the proposed scheme.	
\end{abstract}


\section{Introduction}
The sixth generation (6G) communication networks are envisioned to provide ubiquitous connectivity and intelligent services for machine-type and robotic applications, including industrial automation, scientific exploration, and emergency response~\cite{6G-1,6G-2,JSAC1}. Rather than merely supporting data transmissions, in such scenarios, the networks are expected to enable closed-loop systems that tightly integrate sensing, communication, computing, and control, referred to as sensing-communication-computing-control ($\textbf{SC}^3$) closed loops~\cite{6G-2,JSAC1,wcl}. Specifically, within an $\textbf{SC}^3$ loop, sensors collect environmental information and target status, then transmit them to a computing unit; the computing unit processes the sensing data and generates control commands, which are then conveyed to on-site robots for task execution. The efficient accomplishment of robotic tasks thus hinges on seamless coordination among all components within the $\textbf{SC}^3$ closed loop.

In remote or post-disaster areas, however, terrestrial communication infrastructures may be unavailable or severely damaged. In such scenarios, non-terrestrial networks, including autonomous aerial vehicles (AAVs, also known as UAVs) and satellites, emerge as promising alternatives to support $\textbf{SC}^3$ closed loops~\cite{network,satellite,satellite1}. Nevertheless, the limited on-board energy of UAVs and the constrained battery capacities of field sensors pose significant challenges to sustainable operation, potentially degrading control performance or even causing task interruption. This necessitates energy-aware design principles for $\textbf{SC}^3$ systems in such environments.

Recently, wireless power transfer (WPT) has emerged as a promising solution to overcome the energy limitations of UAVs and Internet of Things (IoT) devices~\cite{WPT1,WPT_survey}. In particular, by exploiting high-power radio-frequency (RF) signals, satellites are envisioned to provide continuous energy supply to $\textbf{SC}^3$
closed loops from space, thereby enabling long-term robotic operations~\cite{WPT2,WPT3}. The feasibility of space-based WPT has been validated through a series of experimental studies. For instance, Caltech has conducted in-orbit experiments demonstrating power beaming from space, with ground stations successfully detecting the transmitted energy~\cite{WPT4}. Meanwhile, in China, the Sun-Chasing Project has developed a full-link, full-system ground demonstration and verification platform for the OMEGA space solar power satellite~\cite{WPT5}. These advancements underscore the potential of satellite-enabled WPT as a forward-looking architecture for future mission-critical $\textbf{SC}^3$ systems. Nevertheless, integrating WPT with $\textbf{SC}^3$ closed loops introduces new challenges. The harvested energy is inherently limited due to the substantial propagation loss over long distances from satellite to ground. Consequently, system resources, including energy, time, computing capability~\cite{jsac2}, and bandwidth~\cite{satellite2}, must be carefully orchestrated to enhance the performance of $\textbf{SC}^3$ closed loops.

In $\textbf{SC}^3$ closed loops, sensing, communication, computing, control, and energy transfer are tightly coupled. These processes share energy and time resources while serving the common objective of enhancing robotic task performance. However, most existing works on WPT focus on optimizing individual communication or computing modules without adequately capturing the interdependencies among $\textbf{SC}^3$ components. Motivated by these limitations, this paper adopts a holistic perspective to investigate the orchestration of communication, computing, and WPT within $\textbf{SC}^3$ closed loops. Specifically, we jointly optimize uplink/downlink transmit power, bandwidth allocation, computing capability, communication-computing time allocation, and WPT energy allocation, with the aim of enhancing overall closed-loop control performance.

\subsection{Related Works}
Networked control systems (NCSs), wherein control components are interconnected via communication networks, have garnered sustained attention as a pivotal integration of control and communications~\cite{NCS0}. From the control perspective, extensive research has characterized how imperfect communication, manifested as packet loss, delay, and limited data rates, impacts control performance, and has accordingly designed control policies under such constraints. For instance, \cite{NCS1} established that stabilizing a control system necessitates the data rate per control cycle to exceed the system's intrinsic entropy rate. The authors in \cite{LQR} derived the minimum average data rate required to achieve a specified linear quadratic regulator (LQR) cost, a standard metric for quantifying control performance in optimal control theory. In \cite{NCS2}, the tracking performance limitation (TPL) and regulation performance limitation (RPL) of NCSs were investigated under communication constraints including network-induced delay, packet loss, and limited bandwidth, with explicit expressions for both TPL and RPL derived. Furthermore, \cite{NCS3} proposed an event-triggered $H_{\infty}$ control strategy for NCSs accounting for network transmission delays.

From the communication perspective, prior works have examined transmission policies and resource allocation strategies with explicit consideration of control performance. Wang \textit{et al.} \cite{NCS4} devised a transmission scheduling policy to minimize infinite-horizon control cost in time-sensitive wireless NCSs. The authors in \cite{single_loop} introduced a novel metric termed closed-loop negentropy (CNE) to quantify closed-loop control performance, and maximized CNE by optimizing time and bandwidth allocation between uplink and downlink within an $\textbf{SC}^3$ closed loop. In \cite{NCS5}, an optimization-theoretic deep reinforcement learning (DRL) framework was proposed for the joint design of control and communication systems, wherein power consumption was minimized through joint optimization of the control sampling period alongside communication blocklength and packet error probability. The authors in \cite{NCS6} jointly optimized bandwidth, transmit time, and computing capability allocation across multiple $\textbf{SC}^3$ closed loops to minimize aggregate LQR cost. These works have underscored the importance of control-communication co-design. However, they typically assume sufficient or fixed energy supply for all devices. In energy-constrained scenarios, the additional coupling introduced by energy harvesting and energy-aware operation remains inadequately addressed.

Recently, WPT has emerged as a promising solution to address energy scarcity in IoT devices. Specifically, by harvesting energy from RF signals transmitted by power stations (PSs) or satellites, sensor operational lifetime can be significantly extended and sensing capabilities enhanced. In \cite{WP-s1}, the authors investigated a wireless-powered sensor network (WPSN) wherein sensors harvest energy from a hybrid access point and transmit information via non-orthogonal multiple access (NOMA). Energy efficiency was maximized through joint optimization of harvesting time and transmit power. The authors in \cite{WP-s2} employed an intelligent reflecting surface (IRS) to assist WPSN operation, maximizing sum throughput via joint optimization of phase shift matrices and transmission time allocation. In \cite{WP-s3}, the authors maximized system sum throughput under two deployment scenarios: one where the PS and sensors belong to the same service operator, and another where they belong to different operators. Nevertheless, these studies focus primarily on sensing and uplink communication, without holistic consideration of complete $\textbf{SC}^3$ closed loops.

Integrated sensing and communication (ISAC), which combines sensing and communication functionalities within a unified system, has emerged as a key enabling technology for 6G networks~\cite{ISAC1}. Recent research has explored the integration of WPT with ISAC. In \cite{WP_ISAC1}, a reconfigurable intelligent surface (RIS)-enabled wireless-powered communication network was investigated for ISAC applications, where the base station performs radar sensing for multiple targets while simultaneously delivering wireless power to energy-constrained devices. The sum rate was maximized through joint optimization of energy beamforming, radar beamforming, RIS phase shifts, and transmission time-slot allocation. The authors in \cite{WP_ISAC2} proposed a UAV-based ISAC architecture wherein multiple radar-equipped UAVs concurrently perform sensing tasks and provide wireless energy to communication users. A multi-objective optimization problem was formulated to enhance sensing and communication performance via joint design of UAV trajectories, radar transmit waveforms, receive filters, time scheduling, and uplink power allocation. Nevertheless, these works primarily focus on optimizing individual sensing and communication metrics, without addressing the closed-loop control performance central to $\textbf{SC}^3$ systems.

Beyond sensing-oriented applications, WPT has been extensively studied in conjunction with mobile edge computing (MEC) to support computation-intensive tasks. In \cite{WP-c1}, the authors investigated a UAV-enabled wireless-powered MEC system wherein the UAV is equipped with an RF energy transmitter and an MEC server to provide energy and computing services to ground users. The weighted sum computation bits were maximized through joint optimization of central processing unit (CPU) frequencies, offloading times, transmit powers, and UAV trajectory. The authors in \cite{WP-c2} proposed a deep reinforcement learning-based framework for offloading decision design in wireless-powered MEC networks, maximizing the weighted sum computation rate. Wang \textit{et al.} \cite{WP-c3} studied a wireless-powered multi-user MEC system where an access point integrates an MEC server with an RF energy transmitter to serve multiple users. Energy consumption was minimized via joint optimization of energy transmit beamforming, computing capability, offloading strategy, and inter-user time allocation. However, these works typically assume unlimited energy supply at the MEC server. In $\textbf{SC}^3$ closed loops, the computing unit is energy-constrained and must compete for energy harvesting with communications of both sensing data and control commands, introducing additional optimization complexity.

Recently, WPT has been investigated for energy-constrained control systems. In \cite{WP_NCS1}, the authors studied a control system wherein a control station transmits energy and control information to multiple control terminals. The maximum control cost across all terminals was minimized through joint optimization of power and time allocation, subject to transmission rate constraints. The authors in \cite{WP_NCS2} incorporated RF energy harvesting into wireless NCSs, formulating a joint energy harvesting, scheduling, and power control problem to maximize system adaptivity. In \cite{WP_NCS3}, the authors investigated optimal uplink-downlink scheduling in energy-harvesting wireless NCSs to minimize the expected cumulative infinite-horizon discounted cost. By formulating the problem as an infinite-horizon discounted Markov decision process, they revealed threshold-based scheduling behavior. Nevertheless, existing works primarily focus on partial resource optimization, without addressing energy allocation among sensing, communication, and computing processes within $\textbf{SC}^3$ closed loops.

To further enhance energy efficiency, simultaneous wireless information and power transfer (SWIPT), which enables concurrent power delivery and information transmission, has garnered significant attention as a potential energy-efficient solution for IoT devices~\cite{SWIPT0}. In \cite{SWIPT1}, the authors investigated three receiver architectures for SWIPT: time switching (TS), power splitting (PS), and on-off power splitting (OPS), deriving optimal transmission strategies to achieve rate-energy tradeoffs. \cite{SWIPT2} employed NOMA to improve spectral efficiency in TS-based SWIPT systems, proposing a joint power allocation and TS control algorithm to maximize system energy efficiency. In \cite{SWIPT3}, the authors investigated an RIS-assisted high-altitude platform (HAP) SWIPT network, utilizing a multi-layer refracting RIS receiver to mitigate severe fading induced by long transmission distances, and formulated a worst-case sum-rate maximization problem. Nevertheless, existing SWIPT research primarily focuses on communication-centric metrics, with limited consideration for closed-loop control requirements.

\subsection{Main Contributions and Organization}
Motivated by aforementioned limitations, this paper investigates a wireless-powered $\textbf{SC}^3$ system comprising multiple $\textbf{SC}^3$ closed loops, wherein a satellite delivers RF energy from space to both field sensors and a UAV-mounted edge information hub (EIH)~\cite{jsac2} to support control tasks. We jointly optimize transmit power, transmission bandwidth, computing capability, time allocation among communication and computing processes, as well as WPT power allocation, to improve closed-loop control performance. An efficient iterative algorithm is proposed to solve the resulting optimization problem, and the optimal solution is analytically characterized for the single-loop special case. The main contributions are summarized as follows.
\begin{itemize}
	\item We investigate a wireless-powered $\textbf{SC}^3$ system wherein a UAV equipped with an EIH assists multiple robots in performing control tasks. A satellite transmits RF energy from space to both field sensors and the EIH. Accounting for coupling among $\textbf{SC}^3$ components as well as the interplay between energy transfer and information transmission/computing, we jointly optimize transmit power, bandwidth allocation, computing capability, time allocation, and WPT power allocation to improve the overall control performance.
	\item We adopt the LQR cost as the control performance metric and formulate a sum LQR cost minimization problem. Due to the non-convexity arising from variable coupling, we introduce auxiliary variables and develop an efficient iterative algorithm based on sequential convex approximation (SCA). For the single-loop special case, we analytically characterize optimal solution properties in energy-limited scenarios to provide insights for practical parameter configuration.
	\item Simulation results demonstrate significant performance gains of the proposed scheme over baseline approaches, validating the superiority of a holistic design perspective for $\textbf{SC}^3$ systems.
\end{itemize}

The remainder of this paper is organized as follows. Section \ref{sec_system} introduces the system model for the considered wireless-powered $\textbf{SC}^3$ closed loops and presents the LQR cost minimization problem formulation. In Section \ref{sec_algorithm}, we recast the optimization problem into a more tractable form and propose an iterative algorithm for its solution. Section \ref{sec_analysis} characterizes the properties of optimal orchestration for the single-loop special case. Simulation results and further discussions are provided in Section \ref{sec_simulation}. Finally, Section \ref{sec_conclusion} concludes the paper.

\section{System Model and Problem Formulation}
\label{sec_system}

\begin{figure} [t]
	\centering
	\includegraphics[width=0.79\linewidth]{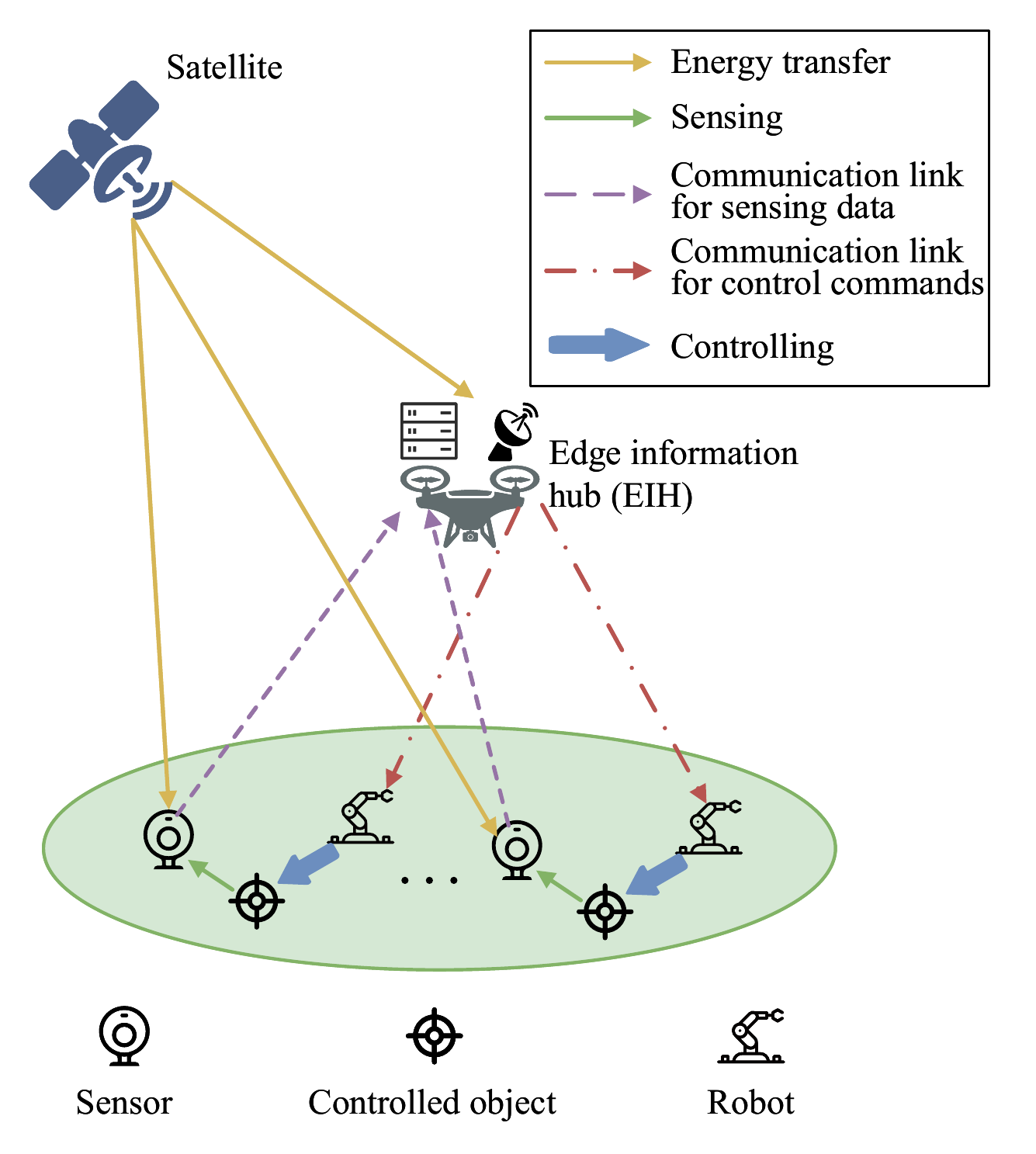}
	\caption{Illustration of a wireless-powered $\textbf{SC}^3$ system: a satellite transmits RF energy from space to multiple field sensors and the EIH; the EIH, acting as a neural center, interconnects sensors and robots to form $\textbf{SC}^3$ closed loops.}
	\label{fig:system}
\end{figure}

As shown in Fig. \ref{fig:system}, we consider a wireless-powered robotic control system comprising $K$ robots performing mission-critical tasks in remote areas, each assisted by a dedicated sensor. A UAV-mounted EIH integrates computing and communication modules to provide edge services for control tasks~\cite{jsac2}. During each control cycle, sensors collect environmental information and target status, transmitting the data to the EIH; the EIH subsequently processes the sensing data and generates control commands, which are then conveyed to field robots to guide their actions. This periodic process constitutes an  $\textbf{SC}^3$ closed loop. Concurrently, a satellite delivers RF energy from space to the $\textbf{SC}^3$
closed loops via wireless power transfer.

\subsection{Sensing and Communication Model}
At each control period, the sensor collects environmental information and target status, containing
both task-related information and redundant contents. We utilize the amount of sensing information per control cycle to measure the sensing capability, denoted as $D^{\text{s}}_k$ for the $k$-th sensor. After sensing, the sensor transmits its obtained data to the EIH for further analysis. The data rate from sensor $k$ to the EIH, or the uplink data rate, can be calculated as
\begin{align}
	R^{\text{u}}_k \left( p^{\text{u}}_k,B^{\text{u}}_k \right)  = B^{\text{u}}_k \log_2 \left(1+\frac{g^{\text{u}}_k p^{\text{u}}_k}{B^{\text{u}}_k  N_0} \right),
\end{align}
where $B^{\text{u}}_k$ denotes the bandwidth for the uplink transmission of sensor $k$, $g^{\text{u}}_k$ represents the channel gain, $p^{\text{u}}_k$ denotes the transmit power, and $N_0$ denotes the channel noise power spectral density.

Denoting the uplink transmission time as $t^{\text{u}}_k$, the data amount from the sensor to the EIH is constrained by the data rate, i.e.,
\begin{equation}
	D^{\text{u}}_k \leq t^{\text{u}}_k R^{\text{u}}_k\left( p^{\text{u}}_k,B^{\text{u}}_k \right).
\end{equation}
In addition, the energy consumption of sensor $k$ per control cycle can be calculated as
\begin{equation}\label{Eu}
	E^{\text{u}}_k = t^{\text{u}}_k p^{\text{u}}_k.
\end{equation}

\subsection{Computing Model}
After receiving the sensing data, the EIH processes the data to extract task-related information, analyzes the system state, and generates the control commands. Denoting the time for the computing process as $t^{\text{c}}_k$, the maximum processed data amount in each control cycle $D^{\text{c}}_k$ can be calculated as
\begin{equation}
	D^{\text{c}}_k \leq \frac{ t^{\text{c}}_k f_k}{\alpha},
\end{equation}
where $\alpha$ denotes the number of CPU cycles for processing the sensing data per bit, and $f_k$ denotes the scheduled computing capability (or CPU frequency) to process the data from sensor $k$.

We assume that each control task is processed by the EIH in a time-division manner, and the total computing time of all tasks should satisfy the constraint
\begin{equation}\label{tc}
	\sum_{k = 1}^K	t^{\text{c}}_k \leq  T,
\end{equation}
where $T$ denotes the time duration of one control cycle.

The energy consumption of the computing process is modeled as\footnote{This model primarily captures dynamic power consumption while neglecting static and leakage components, as dynamic power dominates in practical systems~\cite{MEC_e}. Notably, incorporating a frequency-independent static power term would not affect the convexity analysis or the validity of the proposed algorithm.}
\begin{equation}\label{Ec}
	E^{\text{c}}_k = \kappa t^{\text{c}}_k f_k^3,
\end{equation}
where $\kappa$ denotes the computation energy efficiency coefficient of the EIH. 

The computing process can be regarded as an information extraction process, and we denote the information extraction ratio of the computing process as $\rho$, i.e., only a fraction $\rho$ of the raw sensed data is task-relevant and can be eventually converted into useful control information. Next, the EIH transmits the output control commands to the robot to guide its actions. Similarly as the uplink, the data rate from the EIH to robot $k$, or the downlink data rate, can be calculated as
\begin{align}
	R^{\text{d}}_k\left( p^{\text{d}}_k,B^{\text{d}}_k\right)  = B^{\text{d}}_k \log_2 \left(1+\frac{g^{\text{d}}_k p^{\text{d}}_k}{B^{\text{d}}_k N_0} \right),
\end{align}
where $g^{\text{d}}_k$ denotes the downlink channel gain, $B^{\text{d}}_k$ denotes the downlink transmission bandwidth, and $p^{\text{d}}_k$ denotes the transmit power of the EIH. The corresponding data amount satisfies the following constraint
\begin{equation}
	D^{\text{d}}_k \leq t^{\text{d}}_k R^{\text{d}}_k \left( p^{\text{d}}_k, B^{\text{d}}_k \right),
\end{equation}
where $t^{\text{d}}_k$ denotes the downlink transmission time in loop $k$. And the energy consumption of the downlink communication process is
\begin{equation}\label{Ed}
	E^{\text{d}}_k = t^{\text{d}}_k p^{\text{d}}_k.
\end{equation}

The above process proceeds in sequence and shares the time resource, so we have the following time resource constraint
\begin{equation}
	t^{\text{u}}_k + t^{\text{c}}_k + t^{\text{d}}_k \leq T.
\end{equation}

Furthermore, uplink and downlink transmissions within the same $\textbf{SC}^3$ closed loop share bandwidth in a time-division duplex (TDD) manner, while transmissions across different $\textbf{SC}^3$
	closed loops occupy orthogonal frequency bands to avoid inter-loop interference. Accordingly, the bandwidth constraints are given by
\begin{align}
	&B^{\text{u}}_k \leq B_k, 	B^{\text{d}}_k \leq B_k,\\
	&\sum_{k = 1}^{K} B_k \leq B_{\text{max}},
\end{align}
where $B_k$ denotes the allocated bandwidth for loop $k$, and $B_{\text{max}}$ denotes the maximum available bandwidth for all loops.

\subsection{Control Model}
In this paper, the robotic task is modelled as a linear control process, and the state transition equation of the $k$-th robotic task in cycle $t$ is given by\cite{NCS1}
\begin{equation}\label{system}
	\mathbf{x}_{k,t+1} = \mathbf{\Phi}_k\mathbf{x}_{k,t}+\mathbf{\Gamma}_k\mathbf{u}_{k,t}+\mathbf{v}_{k,t},
\end{equation}
where $\mathbf{x}_{k,t}\in \mathbb{R}^{n_k}$ denotes the system state, $\mathbf{u}_{k,t}\in \mathbb{R}^{m_k}$ denotes control input, and $\mathbf{v}_{k,t}\in \mathbb{R}^{n_k}$ denotes the Gaussian control noise with mean zero and covariance $\mathbf{V}_k$. The matrices $\mathbf{\Phi}_k\in \mathbb{R}^{n_k\times n_k}$ and $\mathbf{\Gamma}_k\in \mathbb{R}^{n_k\times m_k}$ denote the state transition matrix and control input matrix, respectively. 

Inspired by the optimal control theory, the control performance can be measured by the LQR cost, defined as\cite{LQR}
\begin{equation}\label{LQR}
	l_k \triangleq \lim\limits_{N\rightarrow \infty}\mathbb{E} \left[ \frac{1}{N}\sum_{t = 1}^{N} \left(\mathbf{x}_{k,t}^\text{T}\mathbf{Q}_k\mathbf{x}_{k,t} +\mathbf{u}_{k,t}^\text{T}\mathbf{R}_k\mathbf{u}_{k,t}\right) \right],
\end{equation}
where $\mathbf{Q}_k$ and $\mathbf{R}_k$ are weight parameter matrices, which can be set according to practical requirements. The term $\mathbf{x}_{k,t}^\text{T}\mathbf{Q}_k\mathbf{x}_{k,t}$ measures the deviation of the system from desired zero state, and the term $\mathbf{u}_{k,t}^\text{T}\mathbf{R}_k\mathbf{u}_{k,t}$ evaluates the control energy consumption. A smaller LQR cost indicates a better control performance.

According to \cite{LQR}, to achieve an objective LQR cost $l_k$, the average end-to-end information amount transmitted from the sensor to the robot, denoted as $D_k$, must satisfy the following constraint
\begin{equation}\label{CNE_LQR}
	D_k \geq \log_2 |\det \mathbf{\Phi}_k|  + \frac{n_k}{2} \log_2 \left( 1+ \frac{ n_k N \left( \mathbf{v}_k\right)|\det  \mathbf{M}_k|^\frac{1}{n_k} }{l_k-\text{tr}\left( \mathbf{V}_k\mathbf{S}_k\right)} \right),
\end{equation}
where $\log_2 |\det \mathbf{\Phi}_k|$ denotes the intrinsic entropy rate of control system $k$, $N\left( \mathbf{v}_k\right) \triangleq \frac{1}{2\pi e}\exp \left( \frac{2}{n_k} h\left(  \mathbf{v}_k \right)  \right) $ denotes the entropy power of $\mathbf{v}_k$, $h\left(  \mathbf{v}_k \right)$ is its differential entropy, and $\mathbf{M}_k$ and $\mathbf{S}_k$ are the solutions to the Riccati equations about the control parameters, seen in \cite{LQR}. $D_k$ measures the information entropy transmitted through the $k$-th $\textbf{SC}^3$ closed loop per cycle, which is referred to as the CNE~\cite{single_loop}. Specifically, the volume of raw data that can be processed is constrained by the bottleneck among sensing, uplink transmission, and computing capabilities, i.e., $\min\{D_k^{\text{s}},D_k^{\text{u}},D_k^{\text{c}}\}$, where $D_k^{\text{s}}$ denotes the raw sensing data generated by sensor $k$ per control cycle. Since only a fraction $\rho\in(0,1)$ of the raw data is task-relevant and convertible into actionable control information, the effective computing output is given by $\rho\cdot\min\{D_k^{\text{s}},D_k^{\text{u}},D_k^{\text{c}}\}$. Accounting for downlink constraints, the CNE is formulated as
\begin{equation}\label{CNE}
	D_k=\min\left\{\rho D_k^{\text{s}},\,\rho D_k^{\text{u}},\,\rho D_k^{\text{c}},\,D_k^{\text{d}}\right\}.
\end{equation}

\subsection{Energy Transfer Model}
During control task execution, the satellite delivers wireless power to sensors and the EIH via RF signals from space. We assume that the satellite employs multiple beams to transfer energy simultaneously to sensors and the EIH, with energy transfer power for sensor $k$ and the EIH denoted by $P_{k}^{\text{E,s}}$ and $P^{\text{E,e}}$, respectively. The corresponding power constraint is given by
\begin{equation}\label{tec}
\sum_{k = 1}^{K} P_{k}^{\text{E},\text{s}}+P^{\text{E},\text{e}} \leq P_0,
\end{equation}
where $P_0$ denotes the maximum energy transfer power of the satellite.

We assume that WPT and information transmission occupy non-overlapping frequency bands. Moreover, both the sensors and the EIH are equipped with separate antennas for energy harvesting and communications. Consequently, energy harvesting and information transmission/reception can be performed concurrently at both sensors and the EIH, and the mutual interference is negligible.

To maintain the $\textbf{SC}^3$ closed loop, the harvested energy of the sensor and EIH should be larger than the energy consumed per control cycle, and we have the following energy constraints
\begin{align}\label{Econ}
	&\eta^{\text{s}}_k P_{k}^{\text{E},\text{s}} T \geq E^{\text{u}}_k,\\
	&\eta^{\text{e}} P^{\text{E},\text{e}} T \geq \sum_{k = 1}^{K} \left( E^{\text{c}}_k+E^{\text{d}}_k \right) ,
\end{align}
where $\eta^{\text{s}}_k$ and $\eta^{\text{e}}$ denote the energy transfer efficiency from the satellite to sensor $k$ and the EIH, respectively, encompassing the aggregate effects of long-distance propagation attenuation, antenna gains, RF-to-DC conversion efficiency, and other implementation losses. {Specifically, $\eta^{\text{s}}_k$ and $\eta^{\text{e}}$ can be approximately expressed as~\cite{WPT_ref}
\begin{align}
	&\eta^{\text{s}}_k = \eta^{\text{s}}_{\text{b},k}\eta^{\text{s}}_{\text{atm},k}\eta^{\text{s}}_{\text{rf-dc},k}\eta^{\text{s}}_{\text{add},k},\\
	&\eta^{\text{e}} = \eta^{\text{e}}_{\text{b}}\eta^{\text{e}}_{\text{atm}}\eta^{\text{e}}_{\text{rf-dc}}\eta^{\text{e}}_{\text{add}},
	\end{align}
where $\eta^{\text{s}}_{\text{b},k}$ and $\eta^{\text{e}}_{\text{b}}$ denote the beam collection efficiencies, $\eta^{\text{s}}_{\text{atm},k}$ and $\eta^{\text{e}}_{\text{atm}}$ denote the atmospheric transmission efficiencies, $\eta^{\text{s}}_{\text{rf-dc},k}$ and $\eta^{\text{e}}_{\text{rf-dc}}$ represent the RF-to-DC conversion efficiencies, and $\eta^{\text{s}}_{\text{add},k}$ and $\eta^{\text{e}}_{\text{add}}$ represent additional losses such as beam pointing error and polarization mismatch. Moreover, the beam collection efficiencies can be approximated by\cite{WPT_ref}
\begin{align}
	&\eta^{\text{s}}_{\text{b},k} = 1 - \exp \left( -\frac{A_{\text{T}} A^{\text{s}}_{\text{R},k}}{\left(\lambda  d^{\text{s}}_{k}\right) ^2}\right) ,\label{eta_bs}\\
	&\eta^{\text{e}}_{\text{b}} = 1 - \exp \left( -\frac{A_{\text{T}} A^{\text{e}}_{\text{R}}}{\left(\lambda  d^{\text{e}}\right) ^2}\right) ,\label{eta_be}
	\end{align}
where $A_{\text{T}}$ denotes the transmit aperture area, $A^{\text{s}}_{\text{R},k}$ and $A^{\text{e}}_{\text{R}}$ denote the receive aperture areas of the sensors and EIH, $d^{\text{s}}_{k}$ and $d^{\text{e}}$ denote the satellite-to-sensor and satellite-to-EIH distances, and $\lambda$ represents the wavelength of the WPT beam.}

\subsection{Problem Formulation}
In this paper, we aim to minimize the sum LQR cost of the whole system by optimizing the uplink/downlink transmit power $\mathbf{p}^{\text{u/d}}= \left\{p^{\text{u/d}}_k\right\}$, the uplink/downlink transmission bandwidth $\mathbf{B}^{\text{u/d}}= \left\{B^{\text{u/d}}_k\right\}$, the computing frequency  $\mathbf{f} = \left\{f_{k}\right\}$, the time allocation for the uplink, downlink, and computing process, as well as the energy transfer power allocation. The LQR minimization problem is formulated as
\begin{subequations}\label{P1}
	\begin{align} \min\limits_{\substack{\mathbf{f},\mathbf{p}^{\text{u}},\mathbf{B}^{\text{u}},\mathbf{t}^{\text{u}},\mathbf{P}^{\text{E},\text{s}}\\\mathbf{t}^{\text{c}},\mathbf{p}^{\text{d}},\mathbf{B}^{\text{d}},\mathbf{t}^{\text{d}},P^{\text{E},\text{e}}}} \quad &\sum_{k = 1}^K l_k \\
		\text{s.t.} \quad &D_{k} \geq \log_2 |\det \mathbf{\Phi}_k| \nonumber\\ &\qquad +\frac{n_k}{2} \log_2 \left( 1  \! + \! \frac{ n_k N \left(  \! \mathbf{v}_k \! \right)|\det  \mathbf{M}_k|^\frac{1}{n_k} }{l_k-\text{tr}\left( \mathbf{V}_k\mathbf{S}_k\right)} \right),\label{P1b}\\
		&D_k \leq \min \left\{\rho	D^\text{s}_k,\rho	D^{\text{u}}_k,\rho	D^{\text{c}}_k ,	D^{\text{d}}_k \right\},\forall k,\\
		&D^{\text{u}}_k = t^{\text{u}}_k R^{\text{u}}_k \left( p^{\text{u}}_k, B^{\text{u}}_k \right),\forall k,\\
		&D^{\text{c}}_k = \frac{t^{\text{c}}_k f_k}{\alpha },\forall k\\
		&D^{\text{d}}_k = t^{\text{d}}_k R^{\text{d}}_k \left( p^{\text{d}}_k,B^{\text{d}}_k \right),\forall k,\label{P1c}\\
		&B^{\text{u}}_k \leq B_k, 	B^{\text{d}}_k \leq B_k,\\
		&\sum_{k = 1}^{K}B_k \leq B_\text{max},\label{P1h}\\
		&t^{\text{u}}_k + t^{\text{c}}_k + t^{\text{d}}_k \leq T,\forall k,\label{P1i}\\
		&\sum_{k = 1}^{K} P_{k}^{\text{E},\text{s}}+P^{\text{E},\text{e}} \leq P_0,\\
		&\sum_{k = 1}^K	t^{\text{c}}_k \leq  T,\label{P1k}\\
		&t^{\text{u}}_k p^{\text{u}}_k\leq \eta^{\text{s}}_k P_{k}^{\text{E},\text{s}} T ,\forall k, \label{P1l}\\
		&\sum_{k = 1}^K \left( \kappa t^{\text{c}}_k f_k^3+t^{\text{d}}_k p^{\text{d}}_k\right) \leq \eta^{\text{e}} P^{\text{E},\text{e}} T,\label{P1m}\\
		&p^{\text{u}}_k\leq P_{\text{umax}}, p^{\text{d}}_k\leq P_{\text{dmax}}, \label{f12}\\
		&f_k\leq F_{\max}, \label{f13}
	\end{align}
\end{subequations}
where $\mathbf{t}^{\text{u}}= \left\{t^{\text{u}}_k\right\}$, $\mathbf{t}^{\text{c}}= \left\{t^{\text{c}}_k\right\}$, $\mathbf{t}^{\text{d}}= \left\{t^{\text{d}}_k\right\}$, and $\mathbf{P}^{\text{E,s}}= \left\{P^{\text{E,s}}_k\right\}$ represent the corresponding allocation variables,  $P_{\text{umax}}$ and $P_{\text{dmax}}$ denote the maximum uplink and downlink transmit power, respectively, and $F_{\max}$ denotes the maximum CPU frequency of the EIH. The above formulation adopts several idealized assumptions to focus on the essential interaction among WPT, communication, computing, and control in the considered $\textbf{SC}^3$ closed loops. Despite these abstractions, the formulated sum LQR cost minimization problem is still a non-convex optimization problem due to the coupling among the variables across multiple $\textbf{SC}^3$ closed loops. In the next section, we transform it to a more tractable form, and propose an iterative algorithm to solve it.

\section{Problem Transformation and Iterative Approach}
\label{sec_algorithm}
\subsection{Problem Transformation}
First, the bandwidth within each $\textbf{SC}^3$ closed loop should be fully utilized, as improving the bandwidth can achieve high data rate with no increase of energy. Therefore, we have $ B^{\text{u}}_k =B^{\text{d}}_k  =B_k $.

To handle the non-convex coupling between the transmit power and time, we introduce new variables $E^{\text{u}}_k$, $E^{\text{c}}_k$, and $E^{\text{d}}_k$, denoting the energy used in the uplink, computing, and downlink processes of $\textbf{SC}^3$ closed loop. Then the power and computing capability can be expressed in terms of the energy and time based on \eqref{Eu}, \eqref{Ec}, and \eqref{Ed}, as follows
\begin{align}
	&p^{\text{u}}_k = \frac{E^{\text{u}}_k}{t^{\text{u}}_k},\label{puk}\\
	&p^{\text{d}}_k = \frac{E^{\text{d}}_k}{t^{\text{d}}_k},\label{pdk}\\
	&f_k = \left( \frac{E^{\text{c}}_k}{\kappa t^{\text{c}}_k}\right) ^{1/3}.\label{fk}
\end{align}

In addition, the constraints \eqref{P1l} and \eqref{P1m} should hold with equality to avoid the waste of the transferred energy. Therefore, the energy transfer power can be expressed with $E^{\text{u}}_k$, $E^{\text{c}}_k$, and $E^{\text{d}}_k$ as 
\begin{align}
	& P_{k}^{\text{E},\text{s}} = \frac{E^{\text{u}}_k}{T \eta^{\text{s}}_k},\label{tEs}\\
	& P^{\text{E},\text{e}} = \sum_{k = 1}^{K}\frac{E^{\text{c}}_k+E^{\text{d}}_k}{T \eta^{\text{e}}}.\label{tEe}
\end{align}

Based on the above substitution, the original optimization problem \eqref{P1} can be rewritten as
\begin{subequations}\label{P2}
	\begin{align} \min\limits_{\substack{\mathbf{E}^{\text{c}},\mathbf{E}^{\text{u}},\mathbf{B},\mathbf{t}^{\text{u}}\\\mathbf{t}^{\text{c}},\mathbf{E}^{\text{d}},\mathbf{t}^{\text{d}}}} \quad &\sum_{k = 1}^K l_k  \\
		\text{s.t.} \quad &D_k \leq \min \left\{\rho	D^\text{s}_k,\rho	D^{\text{u}}_k,\rho	D^{\text{c}}_k ,	D^{\text{d}}_k \right\},\forall k,\label{P2b}\\
		&D^{\text{u}}_k = t^{\text{u}}_k  B_k \log_2 \left(1+\frac{g^{\text{u}}_k E^{\text{u}}_k}{ t^{\text{u}}_k B_k N_0} \right),\forall k,\label{P2c}\\
		&D^{\text{c}}_k =  t^{\text{c}}_k \left( \frac{E^{\text{c}}_k}{\alpha^3\kappa t^{\text{c}}_k}\right) ^{1/3},\forall k,\label{P2d}\\
		&D^{\text{d}}_k = t^{\text{d}}_k  B_k \log_2 \left(1+\frac{g^{\text{d}}_k E^{\text{d}}_k}{ t^{\text{d}}_k B_k N_0} \right),\forall k,\label{P2e}\\
		&\sum_{k = 1}^{K} \left( \frac{E^{\text{u}}_k}{\eta^{\text{s}}_k}+\frac{E^{\text{c}}_k+E^{\text{d}}_k}{\eta^{\text{e}}}\right)  \leq P_0 T,\label{P2f}\\
		&\frac{E^{\text{u}}_k}{t^{\text{u}}_k}\leq P_{\text{umax}}, \frac{E^{\text{d}}_k}{t^{\text{d}}_k}\leq P_{\text{dmax}}, \label{P2g}\\
		&\left( \frac{E^{\text{c}}_k}{\kappa t^{\text{c}}_k}\right) ^{1/3}\leq F_{\max}, \label{P2h}\\
		&\eqref{P1b},\eqref{P1h}-\eqref{P1i},\eqref{P1k},\nonumber
	\end{align}
\end{subequations}
where $\mathbf{B}= \left\{B_k\right\}$ denotes the bandwidth allocation among $\textbf{SC}^3$ closed loops.

We further introduce slack variables $\mathbf{s}^{\text{u/d}}= \left\{s^{\text{u/d}}_k\right\}$ to handle the bilinear products $t^{\text{u}}_k B_k$ and $t^{\text{d}}_k B_k$ in \eqref{P2c} and \eqref{P2e}, and reformulate the problem in \eqref{P2} as follows
\begin{subequations}\label{P3}
	\begin{align} \min\limits_{\substack{\mathbf{E}^{\text{c}},\mathbf{E}^{\text{u}},\mathbf{B},\mathbf{t}^{\text{u}},\mathbf{s}^{\text{u}}\\\mathbf{t}^{\text{c}},\mathbf{E}^{\text{d}},\mathbf{t}^{\text{d}},\mathbf{s}^{\text{d}}}} \quad &\sum_{k = 1}^K l_k  \\
		\text{s.t.} \quad &D_k \leq \min \left\{\rho	D^\text{s}_k,\rho	D^{\text{u}}_k,\rho	D^{\text{c}}_k ,	D^{\text{d}}_k \right\},\forall k,\label{P3b}\\
		&D^{\text{u}}_k = s^{\text{u}}_k \log_2 \left(1+\frac{g^{\text{u}}_k E^{\text{u}}_k}{ s^{\text{u}}_k N_0} \right),\forall k,\label{P3c}\\
		&s^{\text{u}}_k \leq t^{\text{u}}_k B_k,\label{P3d}\\
		&D^{\text{c}}_k =  t^{\text{c}}_k \left( \frac{E^{\text{c}}_k}{\alpha^3\kappa t^{\text{c}}_k}\right) ^{1/3},\forall k,\label{P3e}\\
		&D^{\text{d}}_k = s^{\text{d}}_k \log_2 \left(1+\frac{g^{\text{d}}_k E^{\text{d}}_k}{ s^{\text{d}}_k N_0} \right),\forall k,\label{P3f}\\
		&s^{\text{d}}_k \leq t^{\text{d}}_k  B_k,\label{P3h}\\
		&\eqref{P1b},\eqref{P1h}-\eqref{P1i},\eqref{P1k},\eqref{P2f}-\eqref{P2h}.\nonumber
	\end{align}
\end{subequations}
The above problem is equivalent to \eqref{P2}, because the data throughput $D^{\text{u}}_k$ and $D^{\text{d}}_k$ are monotonically increasing with respect to $\mathbf{s}^{\text{u}}$ and $\mathbf{s}^{\text{d}}$, respectively, and the constraints \eqref{P3d} and \eqref{P3h} must hold with equality to maximize the CNE $D_k$, therefore minimizing the LQR cost. However, \eqref{P3} is still non-convex. In the next subsection, we propose an iterative algorithm to solve \eqref{P3} based on the SCA method.

\subsection{Iterative Algorithm to Solve \eqref{P3}}
To handle the non-convex constraints \eqref{P3d} and \eqref{P3h}, we first rewrite them in the following equivalent form as
\begin{align}
	&\ln s^{\text{u}}_k - \ln t^{\text{u}}_k - \ln B_k\leq 0,\label{sku}\\
	&\ln s^{\text{d}}_k - \ln t^{\text{d}}_k - \ln B_k\leq 0.\label{skd}
\end{align}

Next, we approximate the left sides of \eqref{sku} and \eqref{skd} with convex functions based on the Taylor expansion. Specifically, the first-order Taylor expansion of $\ln s^{\text{u}}_k$ and $\ln s^{\text{d}}_k$ at the point $s^{\text{u}}_{k,0}$ and $s^{\text{d}}_{k,0}$ are
\begin{align}
	&\frac{s^{\text{u}}_k - s^{\text{u}}_{k,0}}{ s^{\text{u}}_{k,0}}+\ln s^{\text{u}}_{k,0}\geq \ln s^{\text{u}}_k,\label{ineq1}\\
	&\frac{s^{\text{d}}_k - s^{\text{d}}_{k,0}}{ s^{\text{d}}_{k,0}}+\ln s^{\text{d}}_{k,0}\geq \ln s^{\text{d}}_k,\label{ineq2}
\end{align}
where the inequalities are obtained based on the fact that the first-order Taylor expansion of a concave function is its upper bound.

By approximating $\ln s^{\text{u}}_k$ and $\ln s^{\text{d}}_k$ with their Taylor expansions, we propose an iterative algorithm to solve \eqref{P3} based on the SCA method. Specifically, during the $i$-th iteration, we solve the following problem
\begin{subequations}\label{P4}
	\begin{align} \min\limits_{\substack{\mathbf{E}^{\text{c}},\mathbf{E}^{\text{u}},\mathbf{B},\mathbf{t}^{\text{u}},\mathbf{s}^{\text{u}}\\\mathbf{t}^{\text{c}},\mathbf{E}^{\text{d}},\mathbf{t}^{\text{d}},\mathbf{s}^{\text{d}}}} \quad &\sum_{k = 1}^K l_k  \\
		\text{s.t.} \quad &\frac{s^{\text{u}}_k - s^{\text{u}}_{k,i-1}}{ s^{\text{u}}_{k,i-1}}+\ln s^{\text{u}}_{k,i-1} - \ln t^{\text{u}}_k - \ln B_k\leq 0,\label{P4e}\\
		&\frac{s^{\text{d}}_k - s^{\text{d}}_{k,i-1}}{ s^{\text{d}}_{k,i-1}}+\ln s^{\text{d}}_{k,i-1} - \ln t^{\text{d}}_k - \ln B_k\leq 0,\label{P4h}\\
		&\eqref{P1b},\eqref{P1h}-\eqref{P1i},\eqref{P1k},\eqref{P2f}-\eqref{P2h}\nonumber,\\
		&\eqref{P3b}-\eqref{P3c},\eqref{P3e}-\eqref{P3f},\nonumber
	\end{align}
\end{subequations}
where $s^{\text{u}}_{k,i-1}$ and $s^{\text{d}}_{k,i-1}$ denote the values of $s^{\text{u}}$ and $s^{\text{d}}$ by solving \eqref{P4} in the $(i-1)$-th iteration. 

\begin{proposition}\label{prop1}
	The problem \eqref{P4} is a convex optimization problem.
\end{proposition}

\begin{proof}
	First, the constraints in \eqref{P1h}-\eqref{P1i}, \eqref{P1k}, \eqref{P2f} are all affine, and the constraints in \eqref{P2g}–\eqref{P2h} can be rewritten as
	\begin{align}
		&E^{\text{u}}_k\leq P_{\text{umax}}t^{\text{u}}_k,\\
		&E^{\text{d}}_k\leq P_{\text{dmax}}t^{\text{d}}_k,\\
		& E^{\text{c}}_k\leq F_{\max}^3\kappa t^{\text{c}}_k,
	\end{align}
	which are also affine constraints. In addition, it is not difficult to show that the constraints \eqref{P1b}, \eqref{P4e} and \eqref{P4h} are convex due to the concavity of logarithmic functions.
	
	Next, we show that the functions described in \eqref{P3c}, \eqref{P3e}, and \eqref{P3f} are concave. To this end, we utilize a fundamental conclusion in convex optimization theory that the perspective of a concave function is also concave, where the perspective of a function $f\left( x \right) $ is defined as $g \left( t, y \right)  = t f\left( y/t \right)$~\cite{cvx}. We can see that the functions described in \eqref{P3c}, \eqref{P3e}, and \eqref{P3f} are exactly the perspectives of the functions $R^{\text{u}}_k\left( x\right)$, $\frac{1}{\alpha}\left( \frac{x}{\kappa}\right) ^{1/3}$, and $R^{\text{d}}_k\left( x\right)$, respectively. It is not difficult to show that the three original functions are concave functions. Therefore, their perspectives are concave in the corresponding time–energy variables, i.e, the terms $D^{\text{u}}_k$, $D^{\text{c}}_k$, and $D^{\text{d}}_k$ in \eqref{P3b} are concave. As the pointwise minimum of multiple concave functions is also a concave function, we can conclude that the right side of \eqref{P3b} is concave. Based on the above analysis, all the constraints of \eqref{P4} are convex, leading to a convex feasible set of \eqref{P4}. In addition, the objective function of \eqref{P4} is affine. Therefore, \eqref{P4} is a convex optimization problem.
\end{proof}

Based on the above proposition, the problem \eqref{P4} is a convex optimization problem, which can be solved based on the standard convex optimization methods, such as the interior point method. By solving the approximate problem \eqref{P4} iteratively, we can obtain a locally  optimal solution to the original problem \eqref{P1}, which is summarized as \textbf{Algorithm \ref{Algo1}}. The convergence of this algorithm can be demonstrated with the following proposition.

\begin{proposition}\label{prop2}
	The output solution of \textbf{Algorithm \ref{Algo1}} is feasible to the optimization problem in \eqref{P3}. In addition, \textbf{Algorithm \ref{Algo1}} is assured to converge.
\end{proposition}

\begin{proof}
	See Appendix \ref{Appendix1}.
\end{proof}

\begin{algorithm}[t]\label{Algo1}
	\caption{The proposed iterative algorithm for solving problem \eqref{P3}}
	\SetKwInOut{Input}{Input}\SetKwInOut{Output}{Output}\SetKwInOut{Initialize}{Initialization}
	\Input {System parameter $P_{\text{umax}}$, $P_{\text{dmax}}$, $F_{\text{max}}$, $T$, etc; the convergence tolerance $\epsilon$.}
	\Initialize {Calculate a feasible $\mathbf{s}^{\text{u}}_0$ and $\mathbf{s}^{\text{d}}_0$, and set $i = 0$}
	\Repeat{$ \frac{L^{i-1}-L^{i}}{L^{i-1}} <\epsilon$}{
		Set $i = i+1$\;
		Update $\mathbf{s}^{\text{u}}_i$ and $\mathbf{s}^{\text{d}}_i$ by solving \eqref{P4}, denote the value of the objective function as $L^i$\; }
	Calculate the WPT power based on \eqref{tEs} and \eqref{tEe}\;
	Calculate the transmit power and computing capability based on \eqref{puk}-\eqref{fk}\;
	\Output  {the optimal resource allocation results, and the sum LQR cost $L^i$.}
\end{algorithm}

\section{Single-Loop Analysis and Practical Considerations}
\label{sec_analysis}
In this section, we further consider the special case of $K = 1$, i.e., the single-loop case, and analyze the properties of the optimal solution, so as to gain insight for the practical applications. Specifically, for the single-loop case, since the LQR cost is a monotonically decreasing function of the CNE as shown in \eqref{CNE_LQR}, minimizing the LQR cost is equivalent to maximizing the CNE. Therefore, the problem can be reduced to the following form
\begin{subequations}\label{P5}
	\begin{align} \max\limits_{\substack{{E}^{\text{c}},{E}^{\text{u}},{t}^{\text{u}}\\{t}^{\text{c}},{E}^{\text{d}},{t}^{\text{d}}}} \quad &D  \\
		\text{s.t.} \quad &D\leq \rho D^\text{s},\label{P5a}\\
		&D\leq \rho t^{\text{u}}  B_\text{max} \log_2 \left(1+\frac{g^{\text{u}} E^{\text{u}}}{ t^{\text{u}} B_\text{max} N_0} \right),\label{P5b}\\
		&D\leq \rho  t^{\text{c}} \left( \frac{E^{\text{c}}}{\alpha^3\kappa t^{\text{c}}}\right) ^{1/3},\label{P5c}\\
		&D\leq t^{\text{d}}  B_\text{max} \log_2 \left(1+\frac{g^{\text{d}} E^{\text{d}}}{ t^{\text{d}} B_\text{max} N_0} \right),\label{P5d}\\
		&t^{\text{u}} + t^{\text{c}} + t^{\text{d}} \leq T,\label{P5e}\\
		&\frac{E^{\text{u}}}{\eta^{\text{s}}}+\frac{E^{\text{c}}+E^{\text{d}}}{\eta^{\text{e}}}  \leq P_0 T,\label{P5f}\\
		&\frac{E^{\text{u}}}{t^{\text{u}}}\leq P_{\text{umax}}, \frac{E^{\text{d}}}{t^{\text{d}}}\leq P_{\text{dmax}}, \label{P5h}\\
		&\left( \frac{E^{\text{c}}}{\kappa t^{\text{c}}}\right) ^{1/3}\leq F_{\max},\label{P5i}
	\end{align}
\end{subequations}
where the index $k$ is omitted. The above problem is a convex problem and can be solved efficiently by standard convex optimization methods (e.g., the interior-point method). However, such a direct numerical solution provides limited insight into the intrinsic structure of the optimal resource allocation. We consider the energy-constrained scenario where the energy is the bottleneck of the system performance. In such cases, the sensor data throughput constraint \eqref{P5a}, the maximum transmit power constraint \eqref{P5h}, and the computing capability constraints \eqref{P5i} can be omitted. Under this regime, the optimal solution can be expressed based on a unified parameter, which is summarized in the following proposition.

\begin{proposition}\label{prop3}
	Consider problem \eqref{P5} in the energy-limited regime, where the constraints \eqref{P5a}, \eqref{P5h} and \eqref{P5i} are inactive at the optimum. Then, the optimal solution can be expressed in a semi-closed form parameterized by a scalar $\xi>0$.
	
	Defining the uplink and downlink signal-to-noise ratios (SNRs) as
	\begin{align}
		\gamma^{\text{u}} \triangleq \frac{g^{\text{u}}E^{\text{u}}}{B_{\max}N_0 t^{\text{u}}},
		\qquad
		\gamma^{\text{d}} \triangleq \frac{g^{\text{d}}E^{\text{d}}}{B_{\max}N_0 t^{\text{d}}},
	\end{align}
	then the following equations must hold
	\begin{align}
		\frac{\eta^{\text{s}} g^{\text{u}}}
		{\left(1+\gamma^{\text{u}}\right)\ln(1+\gamma^{\text{u}})-\gamma^{\text{u}}}
		=\xi,\label{prop3:eq1}
		\\
		\frac{\eta^{\text{e}} g^{\text{d}}}
		{\left(1+\gamma^{\text{d}}\right)\ln(1+\gamma^{\text{d}})-\gamma^{\text{d}}}
		=\xi.\label{prop3:eq2}
	\end{align}
	The left side of \eqref{prop3:eq1} and \eqref{prop3:eq2} are monotonically decreasing functions of $\gamma^{\text{u}}$ and $\gamma^{\text{d}}$, respectively. Therefore, the corresponding SNRs can be uniquely determined by $\xi$. The above mappings can be denoted by
	\begin{align}
		\gamma^{\text{u}} = \gamma^{\text{u}}\left( \xi\right) , \\
		\gamma^{\text{d}} = \gamma^{\text{d}}\left( \xi\right) ,
	\end{align}
	which are decreasing functions of $\xi$.
	
	Then, the optimal time allocation can be expressed by the CNE $D$ and the scalar $\xi>0$, as
	\begin{align}
		t^{\text{u}} &= \frac{D}{\rho B_{\max}\log_2\left( 1+\gamma^{\text{u}}\left( \xi\right)\right) }, \label{T_closed}\\
		t^{\text{d}} &= \frac{D}{B_{\max}\log_2\left( 1+\gamma^{\text{d}}\left( \xi\right)\right) },\label{Td_closed} \\
		t^{\text{c}} &= \frac{\kappa^{1/3}\alpha}{\rho}
		\left(\frac{2\xi}{\eta^{\text{e}}}\right)^{1/3} D,\label{T_closed1}
	\end{align}
	and the optimal energy allocation is given by
	\begin{align}
		E^{\text{u}} &= \frac{\gamma^{\text{u}}N_0}
		{\rho g^{\text{u}}\log_2\left( 1+\gamma^{\text{u}}\left( \xi\right)\right) } D,\label{Eu_closed} \\
		E^{\text{d}} &= \frac{\gamma^{\text{d}}N_0}
		{g^{\text{d}}\log_2\left( 1+\gamma^{\text{d}}\left( \xi\right)\right) } D,\label{Ed_closed} \\
		E^{\text{c}} &= \frac{\kappa^{1/3}\alpha(\eta^{\text{e}})^{2/3}}
		{2\rho}(2\xi)^{-2/3} D.\label{E_closed}
	\end{align}
	
	Moreover, $\xi$ can be obtained by solving the following one-dimensional equation
	\begin{align}\label{prop3:eq}
		P_0A\left( \xi\right) -B\left( \xi\right)=0,
	\end{align}
	where the functions $A\left( \xi\right) $ and $B\left( \xi\right) $ are defined in \eqref{defA} and \eqref{defB}, as shown at the top of the next page.
	After obtaining the optimal $\xi^\star$, the optimal objective value is
	\begin{align}
		D^\star=\frac{T}{A(\xi^\star)}=\frac{TP_0}{B(\xi^\star)},
	\end{align}
	and all optimal primal variables can be obtained directly from the above expressions.
\end{proposition}

\begin{proof}
	See Appendix \ref{Appendix2}.
\end{proof}

\begin{figure*}[!t]
	\begin{align}
		A(\xi)
		&\triangleq
		\frac{1}{\rho B_{\max}\log_2\left( 1+\gamma^{\text{u}}\left( \xi\right) \right) }
		+\frac{1}{B_{\max}\log_2\left( 1+\gamma^{\text{d}}\left( \xi\right) \right) }
		+\frac{\kappa^{1/3}\alpha}{\rho}
		\left(\frac{2\xi}{\eta^{\text{e}}}\right)^{1/3},\label{defA}\\
		B(\xi)
		&\triangleq
		\frac{\gamma^{\text{u}}\left( \xi\right) N_0}
		{\eta^{\text{s}}\rho g^{\text{u}}\log_2\left( 1+\gamma^{\text{u}}\left( \xi\right) \right) }
		+\frac{\gamma^{\text{d}}\left( \xi\right) N_0}
		{\eta^{\text{e}} g^{\text{d}}\log_2\left( 1+\gamma^{\text{d}}\left( \xi\right) \right) }
		+\frac{\kappa^{1/3}\alpha(\eta^{\text{e}})^{-1/3}}
		{2\rho}(2\xi)^{-2/3}.\label{defB}
	\end{align}
	\hrulefill
\end{figure*}

\begin{remark}
	Since $\xi^\star$ is independent of the control cycle time $T$, varying $T$ only scales the optimal throughput and the absolute resource consumption, while leaving the optimal allocation ratios unchanged.
\end{remark}

\begin{remark}
	By eliminating $\xi$ from \textbf{Proposition \ref{prop3}}, we obtain
	$
	(1+\gamma^{\mathrm{u}})\ln(1+\gamma^{\mathrm{u} })-\gamma^{\mathrm{u} }
	=
	\frac{2\eta^{\mathrm{s}}g^{\mathrm{u}}\kappa}{\eta^{\mathrm{e}}}f^3,
	$
	and
	$
	(1+\gamma^{\mathrm{d} })\ln(1+\gamma^{\mathrm{d} })-\gamma^{\mathrm{d} }
	=
	2g^{\mathrm{d}}\kappa f^3.
	$
	Therefore, both optimal uplink and downlink SNRs increase monotonically with the optimal CPU frequency. This result shows how the channel gains, WPT efficiencies, and computing energy coefficient jointly determine relationship between the optimal communication quality and computing capability.
\end{remark}

\textbf{Proposition \ref{prop3}} shows that, instead of solving the original multi-variable convex problem directly, one only needs to solve a scalar nonlinear equation to obtain the optimal solution to \eqref{P5}. This significantly simplifies both the computation and the interpretation of the optimal solution. Specifically, the following lemma shows that the left side of \eqref{prop3:eq} is increasing with respect to $\xi$. 

\begin{lemma}\label{lemma1}
	The function $A\left( \xi\right)$ is increasing with respect to $\xi$, and the function  $B\left( \xi\right)$ is decreasing with respect to $\xi$. Therefore, the left side of  \eqref{prop3:eq} is increasing with respect to $\xi$.
\end{lemma}

\begin{proof}
	According to \textbf{Proposition \ref{prop3}},  the left-hand sides of \eqref{prop3:eq1} and \eqref{prop3:eq2} are monotonically decreasing functions of $\gamma^{\text{u}}$ and $\gamma^{\text{d}}$, respectively. Therefore, $\gamma^{\text{u}}\left( \xi\right) $ and $\gamma^{\text{d}}\left( \xi\right) $ are both decreasing functions of $\xi$.
	
	The first two terms of $A\left( \xi\right) $ are decreasing with respect to $\gamma^{\text{u}}\left( \xi\right) $ and $\gamma^{\text{d}}\left( \xi\right) $, respectively, indicating that they are increasing with respect to $\xi$. As the third term of $A\left( \xi\right) $	is increasing with respect to $\xi$, we have $A\left( \xi\right) $ is increasing with $\xi$.
	
	To analyze the monotonicity of $B\left( \xi\right) $, we define the function
	\begin{align}
		\psi(x)\triangleq \frac{x}{\log_2(1+x)}, \qquad x>0.
	\end{align}
	Its derivative can be calculated as 
	\begin{align}
		\psi'(x)
		=
		\frac{\log_2(1+x)-\dfrac{x}{(1+x)\ln 2}}
		{\left[\log_2(1+x)\right]^2}
		>0.
	\end{align}
	It is not difficult to show that
	\begin{align}
		\ln(1+x)-\frac{x}{1+x}>0, \qquad x>0.
	\end{align}
	Therefore, $\psi(x)$ is increasing in $x$. Based on the above analysis, the first two terms of $B\left( \xi\right) $ are increasing with respect to $\gamma^{\text{u}}\left( \xi\right) $ and $\gamma^{\text{d}}\left( \xi\right) $, respectively. Since both $\gamma^{\text{u}}\left( \xi\right) $ and $\gamma^{\text{d}}\left( \xi\right) $ are decreasing with $\xi$, we have the first two terms of $B\left( \xi\right) $ are decreasing with $\xi$ . As the third term of $B\left( \xi\right) $ is also decreasing with respect to $\xi$, we have $B(\xi)$ is decreasing with respect to $\xi$.
	
	Based on the above analysis, the function $P_0A(\xi)-B(\xi)$ is increasing with respect to $\xi$. Therefore, binary search can be applied to solve equation \eqref{prop3:eq} efficiently.
\end{proof}

Based on \textbf{Lemma \ref{lemma1}}, \eqref{prop3:eq} can be solved efficiently through binary search. For the proposed algorithm in \textbf{Proposition \ref{prop3}}, the main computational complexity lies in the one-dimensional search over $\xi$. In each iteration, the corresponding $\gamma^{\mathrm{u}}$ and $\gamma^{\mathrm{d}}$ are obtained from two scalar monotone equations \eqref{prop3:eq1} and \eqref{prop3:eq2}, which can also be solved efficiently by one-dimensional search. Therefore, given an accuracy $\varepsilon$, the computational complexity is $O\!\left(\ln^2 \frac{1}{\varepsilon}\right)$. After obtaining $\xi^\star$, the optimal time and energy allocation solutions can be directly recovered from the semi-closed-form expressions \eqref{T_closed}-\eqref{E_closed}. Hence, the proposed method has a low computational complexity as $O\!\left(\ln^2 \frac{1}{\varepsilon}\right)$.

Further, we derive the properties of the optimal resource allocation in the low-SNR regime, as summarized in the following proposition.

\begin{proposition}\label{prop4}
When the WPT transmit power $P_0$ is very small such that the uplink and downlink lie in the low-SNR regime, the optimal energy allocation approximately satisfies
\begin{align}
E^{\text{u}} :E^{\text{d}} \approx \frac{1}{\rho g^{\text{u}} }:\frac{1}{g^{\text{d}} }.\label{E_allo}
\end{align}
In addition, the time allocation satisfies
\begin{align}
	t^{\text{u}} :t^{\text{d}} \approx \frac{1}{\rho \sqrt{\eta^{\text{s}}g^{\text{u}}} }:\frac{1}{\sqrt{\eta^{\text{e}}g^{\text{d}} }}.\label{t_allo}
\end{align}
\end{proposition}
\begin{proof}
	When the energy is extremely limited, both uplink and downlink operate in the low-SNR regime. Using the Taylor expansion of function $\ln \left( 1+x\right)$, the equations \eqref{prop3:eq1} and  \eqref{prop3:eq2} can be approximated as
	\begin{align}
		\frac{2\eta^{\text{s}} g^{\text{u}}}{\left( \gamma^{\text{u}}\right) ^2} 		=\xi,\label{prop4:eq1}	\\
		\frac{2\eta^{\text{e}} g^{\text{d}}}
		{\left( \gamma^{\text{d}}\right) ^2}
		=\xi,\label{prop4:eq2}
	\end{align}
	which indicates that
	\begin{align}
		\frac{\gamma^{\text{u}}}{\gamma^{\text{d}}} =	\sqrt{\frac{\eta^{\text{s}} g^{\text{u}}}
		{\eta^{\text{e}} g^{\text{d}}}}.\label{prop4:eq3}
	\end{align}
	
	Similarly, \eqref{T_closed} and \eqref{Td_closed} can be obtained as
	\begin{align}
	t^{\text{u}} &= \frac{\ln \left( 2 \right) D}{\rho B_{\max} \gamma^{\text{u}}\left( \xi\right)}, \label{prop4:eq4}\\
	t^{\text{d}} &= \frac{\ln \left( 2 \right)D}{B_{\max}\gamma^{\text{d}}\left( \xi\right)}.\label{prop4:eq5} 
\end{align}
	
		Based on \eqref{prop4:eq3}, \eqref{prop4:eq4}, and \eqref{prop4:eq5}, the equation in \eqref{t_allo} can be derived.
		
		In addition, the energy in \eqref{Eu_closed} and \eqref{Ed_closed} can be approximated
		\begin{align}
			E^{\text{u}} &= \frac{\ln \left( 2 \right) N_0}
			{\rho g^{\text{u}} } D, \\
			E^{\text{d}} &= \frac{\ln \left( 2 \right) N_0}
			{g^{\text{d}}} D,
		\end{align}
		which indicates \eqref{E_allo}.
\end{proof}

The proposed framework offers valuable insights into the joint design of WPT, communication, and computing in $\textbf{SC}^3$ closed loops. Nevertheless, several practical challenges remain for real-world deployment. First, the proposed optimization relies on accurate channel state information (CSI). Rapid channel variations due to node mobility would incur significant estimation overhead, potentially degrading real-time performance. Consequently, the results herein serve as performance benchmarks under ideal CSI assumptions, while robust designs accounting for imperfect or statistical CSI constitute an important avenue for future investigation.

Furthermore, space-based WPT systems currently remain at the prototype validation and experimental demonstration stage, with large-scale practical deployment yet to be realized. The construction of extensive ground energy supply infrastructure entails substantial costs, and significant technological advances are prerequisite for practical viability. Nevertheless, space-based WPT may be justified for mission-critical applications, such as emergency rescue, disaster recovery, and remote operations, where infrastructure-independent energy provisioning and operational continuity outweigh cost considerations. Accordingly, the proposed framework should be regarded as a preliminary technological investigation and a theoretically tractable benchmark for future WPT-enabled mission-critical systems.

When applied to large-scale networks, the proposed framework may encounter several practical challenges. First, as the number of devices and $\textbf{SC}^3$ closed loops grows, solving the joint optimization problem in a fully centralized manner becomes computationally prohibitive. Second, multiple $\textbf{SC}^3$ closed loops may need to cooperate to accomplish common tasks, yet modeling the coupling among such cooperative loops remains an open problem. Consequently, developing distributed resource allocation architectures that decompose global optimization into tractable subproblems represents a critical direction for large-scale deployment. Furthermore, large language models (LLMs), which have demonstrated promise in complex multi-agent and heterogeneous systems~\cite{llm1,llm2}, may offer valuable capabilities for coordination and decision-making in cooperative $\textbf{SC}^3$ environments.

\section{Simulation Results and Discussion}
\label{sec_simulation}
In this section, we provide simulation results to evaluate the proposed scheme. Unless specified otherwise, simulation parameters are set to: $K = 5$, $D^{\text{s}}_k = 100$ Mb, $P_{\text{umax}} = 1$ W, $P_{\text{dmax}} = 10$ W, $B_{\text{max}} = 0.1$ MHz, $N_0=-174$ dBm/Hz. The channel gains for sensor-to-EIH and EIH-to-robot links are modeled as $g = g_0 d^{-2}$~\cite{channel}, where $g_0=1.42\times 10^{-4}$ denotes the channel gain at a reference distance of 1 m, and $d$ denotes the distance between the transmitter and receiver, and the link distances are both set to $2$ kilometers in simulations. For the computing-related parameters, we set $\alpha = 1550.7$ cycles/bit as the 95th percentile of random $\alpha$ in the multimedia applications~\cite{com_para1,com_para2}, and the other parameters are set to $\rho = 0.01$, $\kappa = 1\times 10^{-28}$~\cite{MEC_simu} and $F_{\text{max}} = 10$ GHz.

For control parameters, unless specified otherwise, the state matrices $\mathbf{\Phi}_{k}$ are assumed to be $n_k \times n_k$ diagonal matrices with diagonal elements randomly selected in $\left[ -10, 10\right] $, and $n_k$ is randomly selected from $\left[ 20, 50 \right] $. The control system noises are assumed to be independent Gaussian random variables with zero means and covariance matrices $\mathbf{V}_k = \sigma^2_{\text{V},k}\times\mathbf{I}_{n_k}$, and $\sigma^2_{\text{V},k} =\sigma^2_{\text{V}} =  0.01$ unless specified otherwise. The LQR weight matrices are $\mathbf{Q}_k =  \mathbf{I}_{n_k}, \mathbf{R}_k = \mathbf{0}$.

{For the WPT process, we assume that low Earth orbit (LEO) satellites are used for power transmission\footnote{{Although research on space-based solar power satellites often considers geostationary Earth orbit (GEO) satellites, the long propagation distance from GEO to the ground poses challenges such as limited energy collection efficiency and large transmit/receive aperture requirements. To address this issue, existing studies have proposed to deploy a constellation of transmitarrays in lower orbits between the GEO active array and the ground rectenna array, serving as relays to capture and refocus the energy from the GEO active array\cite{WPT2}. However, utilizing LEO satellites for wireless power transmission also introduces additional issues such as beam alignment and tracking, which necessitate further research.}}, at an altitude of $300$ km. The transmit aperture diameter is set to $200$ m, and the frequency of the WPT beam is $10$ GHz. The receive aperture diameters of the sensors and the EIH are set to $0.1$ m and $0.5$ m, respectively. Based on \eqref{eta_bs} and \eqref{eta_be}, the beam collection efficiencies can be calculated as $\eta^{\text{s}}_{\text{b},k} \approx 3.05\times 10^{-6}$ and $\eta^{\text{e}}_{\text{b}} \approx 7.62\times 10^{-5}$. According to ITU-R P.676\cite{ITU}, the zenith atmospheric attenuation at $10$ GHz is approximately $0.05$ dB, corresponding to an atmospheric transmission efficiency of $98.86\%$. Therefore, we set $\eta^{\text{s}}_{\text{atm},k}=\eta^{\text{e}}_{\text{atm}}\approx 98.86\%$. The RF-to-DC conversion efficiencies are set to $\eta^{\text{s}}_{\text{rf-dc},k}=\eta^{\text{e}}_{\text{rf-dc}}=40.52\%$~\cite{WPT5}. In addition, we assume ideal beam pointing and polarization matching, and set $\eta^{\text{s}}_{\text{add},k}=\eta^{\text{e}}_{\text{add}} = 1$. Based on the above settings, the overall energy transfer efficiencies can be calculated as $\eta_k^{\text{s}}=1.22\times 10^{-6}$ and $\eta^{\text{e}}=3.05\times 10^{-5}$. In addition, the maximum WPT power is set to $P_0 = 500$ kW~\cite{WPT_power}, corresponding to maximum received DC powers of approximately $0.61$ W at each sensor and $15.25$ W at the EIH.}

We compare the proposed scheme with two benchmark schemes, which are introduced as follows.
\begin{itemize}
	\item Minimal CNE maximization: Maximizing the minimal CNE among all the  $\textbf{SC}^3$ loops.
	\item Control-oriented scheme with fixed WPT power allocation: Minimizing the sum LQR cost of all the  $\textbf{SC}^3$ loops with equal WPT power allocation.
\end{itemize}

It is worth noting that although SWIPT is commonly studied in wireless power transfer systems, typical SWIPT schemes assume simultaneous information and power transfer over the same link, which is distinct from the scenario considered herein. Consequently, SWIPT-based approaches are not included as primary benchmarks. Additionally, we evaluated a sum-CNE-maximization scheme as a communication-oriented baseline. However, in multi-loop settings, this scheme tends to concentrate excessive resources on a subset of loops, leaving others with insufficient resources to maintain closed-loop stability. This results in unbounded LQR costs for deprived loops. Accordingly, this scheme is excluded from the primary benchmarks presented in the following figures.

\begin{figure} [t]
	\centering
	\includegraphics[width=\linewidth]{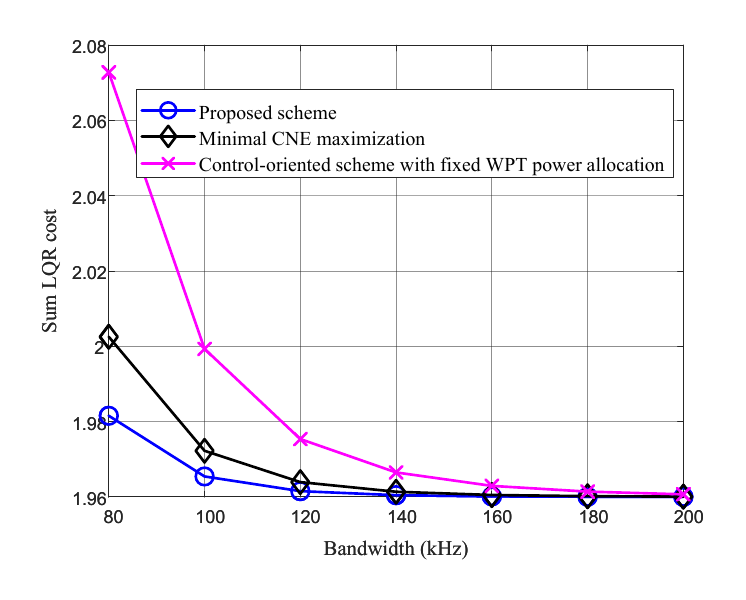}
	\caption{The LQR cost achieved by different schemes, varying with the bandwidth.}
	\label{fig:simu_CNE_vs_B}
\end{figure} 

As expected, the LQR cost decreases with the bandwidth, since the communication components of the $\textbf{SC}^3$ loop can deliver more information when more spectrum is available. Moreover, the proposed scheme attains the lowest LQR cost among the three schemes, showing its superiority. The reason is that the proposed scheme jointly considers all $\textbf{SC}^3$ components together with the energy-transfer process, and directly aims at the overall control performance.

\begin{figure} [t]
	\centering
	\includegraphics[width=\linewidth]{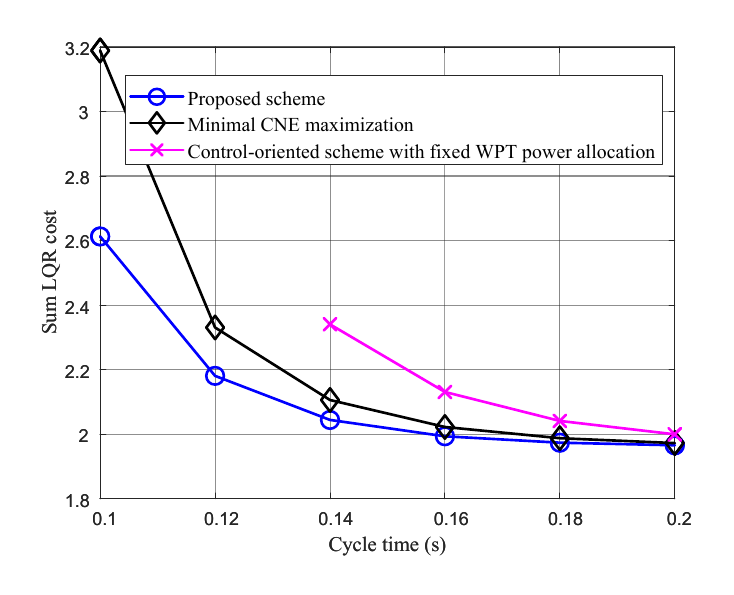}
	\caption{The LQR cost achieved by different schemes, varying with the cycle time.}
	\label{fig:simu_CNE_vs_T}
\end{figure} 

Figure~\ref{fig:simu_CNE_vs_T} shows the LQR cost obtained by different schemes versus the control-cycle time $T$. Again, the proposed scheme consistently achieves the lowest LQR cost, demonstrating its superiority. In particular, when the cycle time is below 0.14 s, some control objects will be unstable with the control-oriented scheme with fixed WPT power allocation, leading to infinite LQR cost.

\begin{figure} [t]
	\centering
	\includegraphics[width=\linewidth]{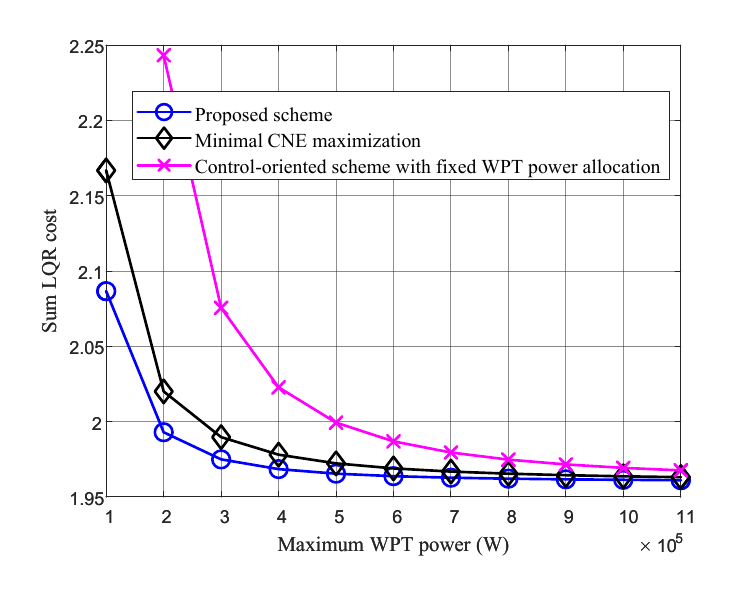}
	\caption{The LQR cost achieved by different schemes, varying with the maximum WPT power.}
	\label{fig:simu_CNE_vs_P}
\end{figure} 

In Fig. \ref{fig:simu_CNE_vs_P}, we show the influence of maximum WPT power (i.e., $P_{0}$) on the sum LQR cost. It can be seen that the LQR costs achieved by the three schemes become similar when the maximum WPT power is high, i.e., $P_{0}$ is larger than 0.8 MW. This is because, with abundant energy, the CNE becomes much larger than $\log_2 |\det \mathbf{\Phi}_k|$. As a result, according to \eqref{CNE_LQR}, the LQR cost approaches its minimum value under ideal communication, namely, $\text{tr}\left( \mathbf{V}_k\mathbf{S}_k\right)$.
\begin{figure} [t]
	\centering
	\includegraphics[width=\linewidth]{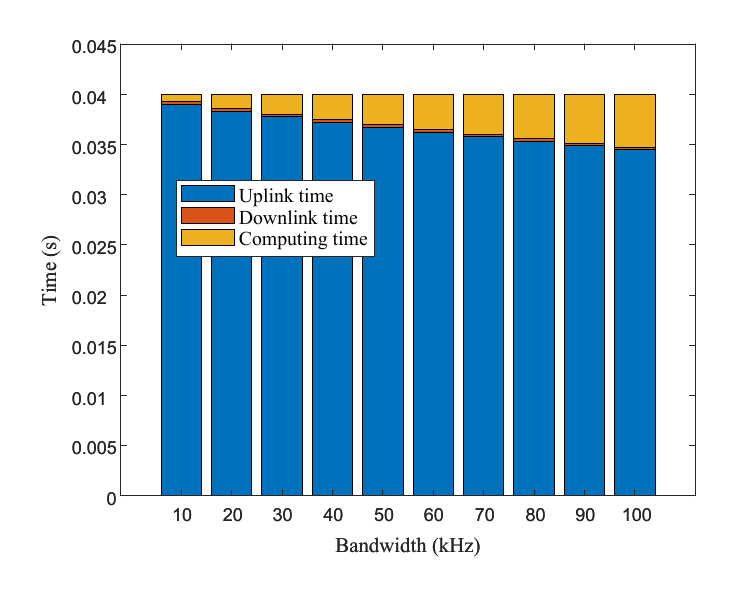}
	\caption{The time allocation results for the uplink/downlink transmissions and computing process, varying with the bandwidth.}
	\label{fig:simu_T_vs_B}
\end{figure}  

Next, we show numerical results for the single-loop case, where the control cycle time is set as $T = 0.04$ s. In Fig. \ref{fig:simu_T_vs_B}, we present the time allocation results for the uplink/downlink transmissions and computing process in the $\textbf{SC}^3$ closed loop with different bandwidth. It can be observed that the uplink transmission time is much larger than the downlink transmission time. This is because the uplink transmits raw sensing data, which typically contains substantial redundancy and therefore requires a longer transmission time. In addition, both the uplink and downlink transmission times decrease as the bandwidth increases. The reason is that a larger bandwidth allows more information to be transmitted per unit time, whereas the volume of data processed per unit time in the computation stage remains unchanged. As a result, the time proportion utilized by communication gradually decreases. This trend is also consistent with the expressions \eqref{T_closed}-\eqref{T_closed1} given in \textbf{Proposition \ref{prop3}}.

\begin{figure} [t]
	\centering
	\includegraphics[width=\linewidth]{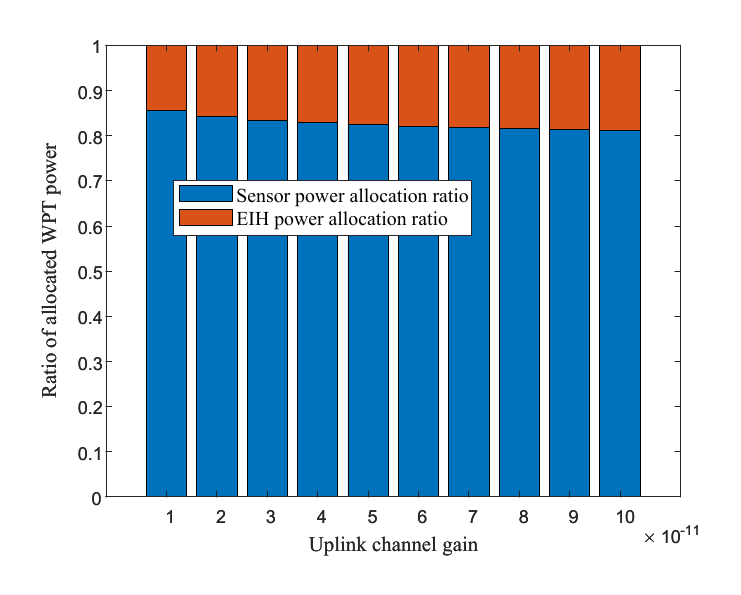}
	\caption{Power allocation results for the WPT process, varying with the uplink channel gain.}
	\label{fig:simu_PE_vs_gu}
\end{figure} 

In Fig. \ref{fig:simu_PE_vs_gu}, we show the WPT power allocation results with different uplink channel gains. It can be seen that the sensor WPT power decreases with the uplink channel gain. This reflects a fairness principle in the closed-loop design: the weaker link should receive more resources. Since the CNE is maximized when the uplink and downlink throughputs are balanced, a sensor with a poor uplink channel needs more harvested energy (i.e., higher WPT power) to raise its uplink rate; when the uplink channel becomes better, this extra WPT power can be shifted to the EIH to keep the balance.

\begin{figure} [t]
	\centering
	\includegraphics[width=\linewidth]{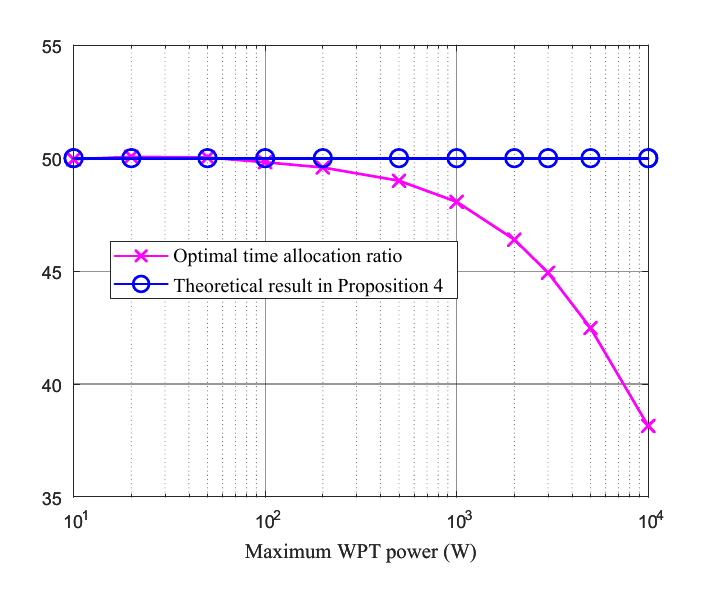}
	\caption{The ratio of allocated uplink transmission time to downlink transmission time, varying with the maximum WPT power.}
	\label{fig:simu_Tratio_vs_P}
\end{figure} 

\begin{figure} [t]
	\centering
	\includegraphics[width=\linewidth]{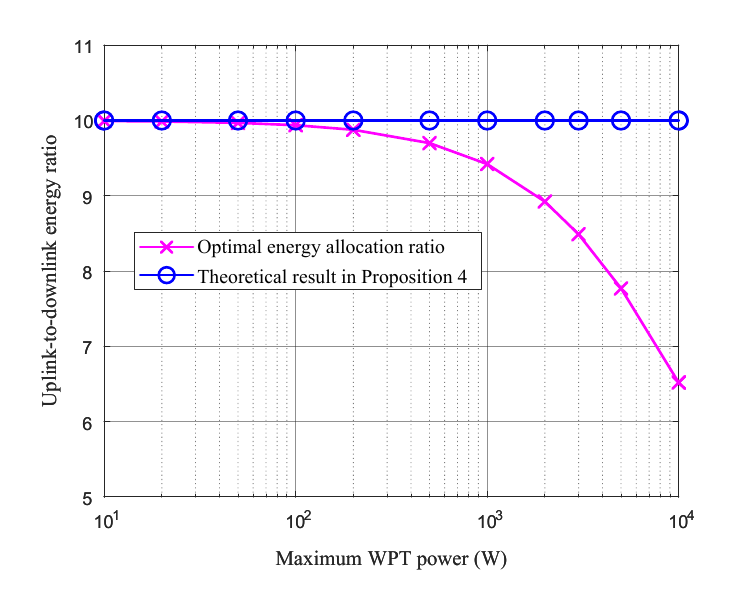}
	\caption{The ratio of allocated energy for the uplink transmission to downlink transmission, varying with the maximum WPT power.}
	\label{fig:simu_Eratio_vs_P}
\end{figure} 

Next, we present simulation results in the energy-limited regime to validate \textbf{Proposition \ref{prop4}}. Fig. \ref{fig:simu_Tratio_vs_P} and Fig. \ref{fig:simu_Eratio_vs_P} depict the uplink-to-downlink time and energy allocation ratio versus the maximum WPT power, where we set $g^{\text{u}} = g^{\text{d}} = 1\times 10^{-15}$ and $\rho = 0.1$. In the low-WPT-power region, the optimal ratios obtained from numerical optimization closely match the theoretical predictions given by \eqref{t_allo} and \eqref{E_allo} in \textbf{Proposition \ref{prop4}}, thereby validating the derived low-SNR approximations. As WPT power increases, the discrepancy between simulated and theoretical results widens, reflecting the breakdown of the low-SNR assumptions underlying \textbf{Proposition \ref{prop4}}. Although the low-SNR condition does not hold across the entire simulation range, \textbf{Proposition \ref{prop4}} nevertheless provides valuable design insights for severely power-constrained scenarios.

\section{Conclusions}
\label{sec_conclusion}
In this paper, we investigated a wireless-powered $\textbf{SC}^3$ system from a holistic system perspective, wherein a satellite delivers RF energy to ground sensors and a UAV-mounted EIH to support closed-loop robotic control. Rather than optimizing communication, computing, and energy transfer in isolation, we explicitly accounted for their intricate coupling and mutual dependencies across the entire $\textbf{SC}^3$
loop. This system-level integration enables us to formulate a unified LQR cost minimization problem that jointly orchestrates transmit powers, bandwidth allocation, computing capability, time allocation for communication and computing, as well as WPT power allocation. An efficient iterative algorithm based on sequential convex approximation was developed to solve this inherently coupled, non-convex optimization problem. For the single-loop special case, we analytically characterized the optimal solution structure and its asymptotic behavior in the energy-limited regime, revealing fundamental insights into how system resources should be balanced under stringent energy constraints. Simulation results demonstrated that the proposed holistic design achieves substantially lower LQR cost compared to benchmark schemes that treat system components separately.

While this paper establishes a comprehensive system-centric framework for modeling, analysis, and optimization in wireless-powered $\textbf{SC}^3$ systems, experimental validation through proof-of-concept implementation remains essential for assessing practical feasibility, which is an important direction for future work.

\appendices
\section{Proof of Proposition \ref{prop2}}\label{Appendix1}
Denoting the optimal solution to problem \eqref{P4} in the $i$-th iteration as $\mathbf{x}_i^* = \left( \mathbf{E}^{\text{c}}_i,\mathbf{E}^{\text{u}}_i,\mathbf{B}_i,\mathbf{t}^{\text{u}}_i,\mathbf{s}^{\text{u}}_i,\mathbf{t}^{\text{c}}_i,\mathbf{E}^{\text{d}}_i,\mathbf{t}^{\text{d}}_i,\mathbf{s}^{\text{d}}_i\right) $, we first prove that $\mathbf{x}_i^*$ is also a feasible solution to \eqref{P3}. Based on the inequality \eqref{ineq1} and constraint \eqref{P4e}, we have
\begin{subequations}\label{AppendixA_1}
	\begin{align}
		&\ln s^{\text{u}}_{k,i}-\ln t^{\text{u}}_{k,i} - \ln B_{k,i}\\
		\leq& \frac{s^{\text{u}}_{k,i} - s^{\text{u}}_{k,i-1}}{ s^{\text{u}}_{k,i-1}}+\ln s^{\text{u}}_{k,i-1} - \ln t^{\text{u}}_{k,i} - \ln B_{k,i}\label{AppendixA_1_2}\\
		\leq& 0,\label{AppendixA_1_3}
	\end{align}
\end{subequations}
which is equivalent to $s^{\text{u}}_{k,i}\leq t^{\text{u}}_{k,i} B_{k,i}$. Similarly, we can prove that  $s^{\text{d}}_{k,i}\leq t^{\text{d}}_{k,i} B_{k,i}$. Therefore, the constraints \eqref{P3d} and \eqref{P3h} hold for the point $\mathbf{x}_i^*$. As the other constraints of problem \eqref{P3} are also  satisfied in \eqref{P4}, we have $\mathbf{x}_i^*$ is also a feasible solution to \eqref{P3}.

Next, we show the convergence of \textbf{Algorithm \ref{Algo1}}. Specifically, based on \eqref{AppendixA_1}, we have
\begin{subequations}
	\begin{align}
		&\frac{s^{\text{u}}_{k,i} - s^{\text{u}}_{k,i}}{ s^{\text{u}}_{k,i}}+\ln s^{\text{u}}_{k,i} - \ln t^{\text{u}}_{k,i} - \ln B_{k,i}\\
		=&\ln s^{\text{u}}_{k,i} - \ln t^{\text{u}}_{k,i} - \ln B_{k,i}\\
		\leq \ &0,\label{AppendixA_2_3}
	\end{align}
\end{subequations}
which indicates that the point $\mathbf{x}_i^*$ also satisfies the constraint \eqref{P4e} of \eqref{P4} in the $(i+1)$-th iteration.  Similarly, $\mathbf{x}_i^*$ satisfies the constraint \eqref{P4h} of \eqref{P4} in the $(i+1)$-th iteration. Therefore, $\mathbf{x}_i^*$ is also feasible to the optimization problem \eqref{P4} in the $(i+1)$-th iteration, indicating that $L^{i}$ is also an achievable objective function value in the $(i+1)$-th iteration. As $\mathbf{x}_{i+1}^*$ minimizes objective function of \eqref{P4} in the $(i+1)$-th iteration, we have $L^{i+1}\leq L^{i}$. Therefore, the objective function value is non-increasing along the iteration, indicating that \textbf{Algorithm \ref{Algo1}} is assured to converge.

\section{Proof of Proposition \ref{prop3}}
\label{Appendix2}
As the problem \eqref{P5} is convex problem, we analyze the optimal solution through Lagrange duality. Let $\mu_1$, $\mu_2$, and $\mu_3$ denote the Lagrange multipliers associated with the constraints \eqref{P5b}, \eqref{P5c}, and \eqref{P5d}, respectively; let $\lambda$ be the multiplier associated with the time constraint \eqref{P5e}; and let $\nu$ be the multiplier associated with the energy constraint \eqref{P5f}. Then, the Lagrangian function of the \eqref{P5} can be written as follows
\begin{align}
	L
	=\,  D
	-\mu_1 \bigl[D-\rho D^{\text{u}}(t^{\text{u}},E^{\text{u}})\bigr]
	-\mu_2 \bigl[D-D^{\text{d}}(t^{\text{d}},E^{\text{d}})\bigr]\nonumber\\
	-\mu_3 \bigl[D-\rho D^{\text{c}}(t^{\text{c}},E^{\text{c}})\bigr] 
	-\lambda (t^{\text{u}}+t^{\text{c}}+t^{\text{d}}-T)\nonumber\\
	-\nu \left(\frac{E^{\text{u}}}{\eta^{\text{s}}}+\frac{E^{\text{d}}}{\eta^{\text{e}}}+\frac{E^{\text{c}}}{\eta^{\text{e}}}-TP_0\right),
\end{align}
where we define
\begin{align}
	&D^{\text{u}}\left( t^{\text{u}},E^{\text{u}}\right)  = t^{\text{u}}  B_\text{max} \log_2 \left(1+\frac{g^{\text{u}} E^{\text{u}}}{ t^{\text{u}} B_\text{max} N_0} \right),\\
	&D^{\text{d}}\left( t^{\text{d}},E^{\text{d}}\right)  = t^{\text{d}}  B_\text{max} \log_2 \left(1+\frac{g^{\text{d}} E^{\text{d}}}{ t^{\text{d}} B_\text{max} N_0} \right),\\
	&D^{\text{c}}\left( t^{\text{c}},E^{\text{c}}\right) = \frac{1}{\alpha} t^{\text{c}} \left( \frac{E^{\text{c}}}{\kappa t^{\text{c}}}\right) ^{1/3}.
\end{align}
By differentiating with respect to $E^{\text{u}}$, $E^{\text{d}}$, and $E^{\text{c}}$, we can obtain
\begin{align}
	\nu
	= \eta^{\text{s}} \mu_1 \rho \frac{\partial D^{\text{u}}}{\partial E^{\text{u}}}
	= \eta^{\text{e}} \mu_2 \frac{\partial D^{\text{d}}}{\partial E^{\text{d}}}
	= \eta^{\text{e}} \mu_3 \rho \frac{\partial D^{\text{c}}}{\partial E^{\text{c}}}.
\end{align}
In addition, differentiating the Lagrangian function with respect to $t^{\text{u}}$, $t^{\text{d}}$, and $t^{\text{c}}$ yields
\begin{align}
	\lambda
	= \mu_1 \rho \frac{\partial D^{\text{u}}}{\partial t^{\text{u}}}
	= \mu_2 \frac{\partial D^{\text{d}}}{\partial t^{\text{d}}}
	= \mu_3 \rho \frac{\partial D^{\text{c}}}{\partial t^{\text{c}}}.
\end{align}

By dividing the corresponding expressions above, the multipliers $\mu_1$, $\mu_2$, and $\mu_3$ can be eliminated, and we can obtain the following equations
\begin{align}\label{KKT1}
	\eta^{\text{s}}
	\frac{\partial D^{\text{u}}/\partial E^{\text{u}}}
	{\partial D^{\text{u}}/\partial t^{\text{u}}}
	=
	\eta^{\text{e}}
	\frac{\partial D^{\text{d}}/\partial E^{\text{d}}}
	{\partial D^{\text{d}}/\partial t^{\text{d}}}
	=
	\eta^{\text{e}}
	\frac{\partial D^{\text{c}}/\partial E^{\text{c}}}
	{\partial D^{\text{c}}/\partial t^{\text{c}}}
	=
	\frac{\nu}{\lambda}.
\end{align}
We define
\begin{align}
	\xi \triangleq \frac{\nu}{\lambda},
\end{align}
then, at the optimum, the uplink transmission, downlink transmission, and computation components share the same effective marginal energy-time tradeoff characterized by $\xi$.

Next, by calculating the derivatives in \eqref{KKT1}, we can obtain the following equations
\begin{align}
	\frac{\eta^{\text{s}} g^{\text{u}}}
	{\left(1+\gamma^{\text{u}}\right)\log\!\left(1+\gamma^{\text{u}}\right)-\gamma^{\text{u}}}
	=
	\frac{\eta^{\text{e}} g^{\text{d}}}
	{\left(1+\gamma^{\text{d}}\right)\log\!\left(1+\gamma^{\text{d}}\right)-\gamma^{\text{d}}}
	\nonumber\\
	=
	\frac{\eta^{\text{e}} t^{\text{c}}}{2E^{\text{c}}}
	=
	\xi,\label{Appendix2_equation}
\end{align}
where we define the uplink and downlink SNRs as
\begin{align}\label{Appendix2_SNR}
	\gamma^{\text{u}} \triangleq \frac{g^{\text{u}} E^{\text{u}}}{B_{\text{max}} N_0 t^{\text{u}}},
	\qquad
	\gamma^{\text{d}} \triangleq \frac{g^{\text{d}} E^{\text{d}}}{B_{\text{max}} N_0 t^{\text{d}}},
\end{align}
This aligns with the equations in \eqref{prop3:eq1} and \eqref{prop3:eq2}.

It is not difficult to prove that the function $\left( 1+x \right) \ln \left( 1+x \right)  -x$ is increasing with respect to $x$. Therefore, the first two terms are monotonically decreasing functions of $\gamma^{\text{u}}$ and $\gamma^{\text{d}}$, respectively. For any given $\xi$, the corresponding SNRs can be uniquely determined, which can be denoted by
\begin{align}
	\gamma^{\text{u}} = f^{\text{u}}(\xi), \\
	\gamma^{\text{d}} = f^{\text{d}}(\xi).
\end{align}

It is worth noting that the data throughput should be balanced at the optimum of \eqref{P5}. Otherwise, the time or energy resources will be wasted. Therefore, the constraints \eqref{P5b}-\eqref{P5d} should hold with equality. Then, based on \eqref{Appendix2_equation} \eqref{Appendix2_SNR}, and \eqref{P5b}-\eqref{P5d}, the time allocation variables can be expressed with $D$ and $\xi$ as
\begin{align}
	t^{\text{u}} &= \frac{D}{B_\text{max}\rho \log\!\left(1+\gamma^{\text{u}}\right)}, \\
	t^{\text{d}} &= \frac{D}{B_\text{max} \log\!\left(1+\gamma^{\text{d}}\right)}, \\
	t^{\text{c}} &= \frac{\kappa^{1/3}\alpha}{\rho}
	\left(\frac{2\xi}{\eta^{\text{e}}}\right)^{1/3} D.
\end{align}
and, the energy allocation variables are
\begin{align}
	E^{\text{u}} &= \frac{\gamma^{\text{u}} N_0}{\rho g^{\text{u}} \log\!\left(1+\gamma^{\text{u}}\right)} D, \\
	E^{\text{d}} &= \frac{\gamma^{\text{d}} N_0}{g^{\text{d}} \log\!\left(1+\gamma^{\text{d}}\right)} D, \\
	E^{\text{c}} &= \frac{\kappa^{1/3}\alpha\left(\eta^{\text{e}}\right)^{2/3}}{2\rho}
	(2\xi)^{-2/3} D.
\end{align}

Substituting the above expressions into the total time constraint and the total energy constraint, respectively, we can obtain the following equations
\begin{align}
	\frac{1}{B_\text{max}\rho \log\!\left(1+\gamma^{\text{u}}\right)}
	+\frac{1}{B_\text{max} \log\!\left(1+\gamma^{\text{d}}\right)}
	+\frac{\alpha}{\rho}
	\left(\frac{2\kappa\xi}{\eta^{\text{e}}}\right)^{1/3}
	=
	\frac{T}{D},
\end{align}
and
\begin{align}
	&\frac{\gamma^{\text{u}} N_0}
	{\eta^{\text{s}} \rho g^{\text{u}} \log\!\left(1+\gamma^{\text{u}}\right)}
	+\frac{\gamma^{\text{d}} N_0}
	{\eta^{\text{e}} g^{\text{d}} \log\!\left(1+\gamma^{\text{d}}\right)}\nonumber
	\\
	&\qquad\qquad\qquad +\frac{\kappa^{1/3}\alpha\left(\eta^{\text{e}}\right)^{-1/3}}{2\rho}
	(2\xi)^{-2/3}
	=
	\frac{TP_0}{D}.
\end{align}
Define the left-hand sides of the above two equations as $A(\xi)$ and $B(\xi)$, respectively. Then we have
\begin{align}
	A(\xi)=\frac{T}{D}, \qquad
	B(\xi)=\frac{TP_0}{D},
\end{align}
which further implies
\begin{align}
	P_0A(\xi)-B(\xi)=0.
\end{align}
This completes the proof.


\newpage

 


\vspace{11pt}


\vfill

\end{document}